\documentclass[11pt]{article}
\pdfoutput=1

\usepackage{hyperref}
\usepackage{fullpage}
\usepackage{amssymb}
\usepackage{mathtools}
\usepackage{amsthm}
\usepackage{amsmath}
\usepackage{eucal}
\usepackage{dsfont}
\usepackage[T1]{fontenc}
\usepackage{lmodern}
\usepackage[stretch=10,shrink=10]{microtype}
\usepackage{color,soul}
\usepackage{extarrows}
\usepackage{overpic}
\usepackage{tikz-cd}
\allowdisplaybreaks[1]

\def\D{\CMcal{D}}
\def\E{\CMcal{E}}

\def\H{\CMcal{H}}

\def\O{\CMcal{O}}

\def\U{\CMcal{U}}

\def\X{\CMcal{X}}
\def\Y{\CMcal{Y}}

\theoremstyle{plain}
\newtheorem{theorem}{Theorem}[section]
\newtheorem{lemma}[theorem]{Lemma}

\theoremstyle{definition}
\newtheorem{definition}[theorem]{Definition}
\newtheorem{claim}[theorem]{Claim}

\newtheorem{fact}[theorem]{Fact}

\newcommand {\minusspace} {\: \! \!}
\newcommand {\smallspace} {\: \!}
\newcommand {\fn} [2] {\ensuremath{ #1 \minusspace \br{ #2 } }}
\newcommand {\Fn} [2] {\ensuremath{ #1 \minusspace \Br{ #2 } }}

\newcommand {\set} [1] {\ensuremath{ \left\lbrace #1 \right\rbrace }}

\newcommand {\br} [1] {\ensuremath{ \left( #1 \right) }}
\newcommand {\Br} [1] {\ensuremath{ \left[ #1 \right] }}

\newcommand {\norm} [1] {\ensuremath{ \left\| #1 \right\| }}
\newcommand {\normsub} [2] {\ensuremath{ \norm{#1}_{#2} }}
\newcommand {\onenorm} [1] {\normsub{#1}{1}}

\newcommand {\abs} [1] {\ensuremath{ \left| #1 \right| }}
\newcommand {\fidelity} [2] {\mathrm{F}\br{#1,#2}}
\newcommand {\bra} [1] {\ensuremath{ \left\langle #1 \right| }}
\newcommand {\ket} [1] {\ensuremath{ \left| #1 \right\rangle }}
\newcommand {\ketbratwo} [2] {\ensuremath{ \left| #1 \middle\rangle \middle\langle #2 \right| }}
\newcommand {\ketbra} [1] {\ketbratwo{#1}{#1}}

\newcommand {\defeq} {\ensuremath{ \stackrel{\mathrm{def}}{=} }}
\newcommand {\mutinf} [2] {\fn{\mathrm{I}}{#1 \smallspace : \smallspace #2}}
\newcommand {\condmutinf} [3] {\mutinf{#1}{#2 \smallspace \middle\vert \smallspace #3}}
\newcommand {\prob} [1] {\Fn{\Pr}{#1}}
\newcommand {\relent} [2]  {\fn{\mathrm{D}}{#1 \middle\| #2}}

\DeclareMathOperator*{\bigE}{\mathbb{E}}
\newcommand {\expec} [2] {\Fn{\bigE_{\substack{#1}}}{#2}}

\newcommand {\Tr} {\ensuremath{ \mathrm{Tr} }}

\newcommand {\id} {\ensuremath{\mathds{1}}}

\newcommand {\email} [1] {\href{mailto:#1}{\texttt{#1}}}

\newcommand {\mytitle} {}
\newcommand {\suppress} [1] {}

\newcommand {\myytitle} {Exponential Separation of Quantum Communication and Classical Information}
\newcommand {\Dave}   {Dave Touchette}
\newcommand {\Anurag}  {Anurag Anshu}
\newcommand {\Penghui} {Penghui Yao}
\newcommand {\Nengkun} {Nengkun Yu}

\newcommand {\CQT} {Centre for Quantum Technologies}
\newcommand {\NUS} {National University of Singapore.}
\newcommand {\CNO} {Department of Combinatorics and Optimization}

\newcommand {\IQCO} {Institute for Quantum Computing}
\newcommand {\UW} {Unversity of Waterloo}
\newcommand {\PI} {Perimeter Institute}
\newcommand{\CQSI}{Centre for Quantum Software and Information,
   Faculty of Engineering and Information Technology, University of
   Technology Sydney}
\newcommand{\QUICS}{Joint Center for Quantum Information and Computer Science, University of Maryland
}
\newcommand {\authorblock} [3] {
	\begin{minipage}[t]{0.3\linewidth}
		\centering
		{#1}\\[0.8ex]
		{\footnotesize {#2}\\[-0.7ex]
		\email{#3}}
	\end{minipage}\vspace{1ex}
}

\hypersetup{
	pdfstartview={FitH},
	pdfdisplaydoctitle={true},
	breaklinks={true},
	bookmarksopen={true},
	bookmarksnumbered={false},
	pdftitle={\mytitle},
}

\begin{document}

\begin{titlepage}
\title{\textbf{\myytitle}\\[2ex]}

\author{
    \authorblock{\Anurag}{\CQT, \NUS}{a0109169@u.nus.edu}
	\authorblock{\Dave}{\IQCO,~and~\CNO,~\UW,~and~\PI}{touchette.dave@gmail.com}\\
	\authorblock{\Penghui}{\QUICS}{phyao1985@gmail.com}
	\authorblock{\Nengkun}{\CQSI}{nengkunyu@gmail.com}
}

\clearpage\maketitle
\thispagestyle{empty}

\begin{abstract}

We exhibit a Boolean function for which the {\em quantum communication complexity} is exponentially larger than the  {\em classical information complexity}.
An exponential separation in the other direction was already known from the work of Kerenidis et.~al.~[SICOMP 44, pp. 1550--1572],
hence our work implies that these two complexity measures are incomparable.
As classical information complexity is an upper bound on {\em quantum information
	complexity}, which in turn is equal to {\em amortized quantum communication complexity}, our work implies that a tight direct sum result
for distributional quantum communication complexity cannot hold.
The function we use to present
such a separation is the \textsf{ Symmetric $k$-ary Pointer Jumping} function introduced by Rao and Sinha~[ECCC TR15-057], whose classical communication complexity is exponentially larger than its classical information complexity.
In this paper, we show that the quantum communication complexity of this function
is polynomially equivalent to its classical communication complexity.
The high-level idea behind our proof is arguably the simplest so far for such an
exponential separation between information and communication, driven by a sequence
of round-elimination arguments, allowing us to simplify further the approach of Rao and Sinha.

As another application of the techniques that we develop, we give a simple
proof for an optimal trade-off between Alice's and Bob's communication 
while computing the
related \textsf{Greater-Than} function on $n$ bits: say Bob communicates at
most $b$ bits, then Alice must send $\frac{n}{2^{O (b)}}$  bits to Bob.
This holds even when allowing pre-shared entanglement.
We also present a \emph{classical} protocol achieving this bound.
\end{abstract}
\end{titlepage}


\section{Introduction}
Communication complexity is a core topic of
computational complexity which studies the number of bits that the
participants in a communication protocol need to
exchange in order to accomplish a distributed task. Designing
generic lower bound methods for communication complexity has been
a central endeavor since the birth of this subject, see~\cite{Kushilevitz96,LeeS:2007} as excellent surveys. One of the most powerful lower
bound methods for {\em randomized communication complexity} (RCC) is {\em information complexity} (IC) introduced
in~\cite{Chakrabarti:2001,Bar-Yossef2002a, Barak:2010:CIC:1806689.1806701}, which studies
the amount of information about the inputs that the players need
to reveal in order to accomplish a communication task.
Investigations of information complexity have led to
numerous elegant compression protocols, which in
turn have led to direct sum and direct product results~\cite{Jain:2003:DST:1759210.1759242,Barak:2010:CIC:1806689.1806701,Braverman2011,Jain:2012,Jain:2012c,Braverman2012,Jain:2011,BravermanRWY:2013,BravermanWY:2013,BravermanW:2015} (and many other works).

The notion of information complexity appears in two flavors. The first is termed  {\em external information complexity}, introduced by Chakrabarti, Shi, Wirth and Yao ~\cite{Chakrabarti:2001}, which measures the amount of information about the inputs that the players reveal to an external observer in the protocol. Formally,
it is defined as $\mutinf{XY}{MR}$, the mutual information between $XY$ and $MR$, where $XY$ is the joint input to the players (with respect to an implicit prior distribution $\mu$); $M$
is the set of messages exchanged in the protocol and $R$ is the public coins shared between the players. The second notion is that of {\em (internal)
information complexity}, formally introduced by Barak, Braverman, Chen and Rao in~\cite{Barak:2010:CIC:1806689.1806701} (building on a related notion introduced by Bar-Yossef, Jayram, Kumar and Sivakumar~\cite{Bar-Yossef2002a}) and defined as $\condmutinf{X}{MR}{Y}+\condmutinf{Y}{MR}{X}$.

Following is a central question in the field of communication complexity. Given a communication protocol with external information complexity $I^{\text{ext}}$, information complexity $I$ and communication complexity $C$, where $I^{\text{ext}}, I\ll C$, can this protocol be simulated by another communication protocol (or {\em compressed}) with a much smaller amount of communication? After a decade's efforts, it is now
known that any such protocol can be compressed to one with communication
complexity $2^{\O\br{I}}$~\cite{Braverman2012}; $\O\br{\sqrt{IC}\log C}$~\cite{Barak:2010:CIC:1806689.1806701}; $\O\br{I^{\text{ext}}\log^2 C}$~\cite{Barak:2010:CIC:1806689.1806701}. For $r$-round protocols,
it can be compressed to $I+\O\br{\sqrt{rI}+r}$~\cite{Braverman2011}.
If the distribution of the input is product, then recent results show that the protocol can be simulated by another protocol with communication complexity $\O\br{I^2\text{polylog} I}$ due to Kol~\cite{Kol:2016:ICP:2897518.2897537}, and to $\O\br{I\log^2 I}$ in a later improvement by Sherstov~\cite{Sherstov:2016}.

An immediate question towards this line of research is whether compressing to $\O\br{I}$, or even $\text{poly} (I)$,
is possible in general. This question is tightly connected to the direct sum question, and
unfortunately, the answer is negative. In a sequence of breakthrough
works, Ganor, Kol and Raz~\cite{GanorKR:2014,GanorKR:2015} exhibited a function
with information complexity $I$ that requires $2^{\Omega\br{I}}$ communication to solve with constant error, say, $1/3$.
A significantly simpler proof was later given by Rao and Sinha~\cite{RaoS:2015}.
These works imply that Braverman's exponential simulation theorem~\cite{Braverman2012}
is tight in some cases. Moreover, since information complexity is equal to
 {\em amortized communication complexity}~\cite{Braverman2011},
this also proves that a tight direct sum result for distributional communication
complexity is not possible in general, resolving a longstanding open problem in
communication complexity.

Much work has been devoted to seeking the quantum analog of information
complexity, inspired by the numerous successful applications of information complexity in classical communication complexity. A major obstacle towards extending the notion of information complexity to the quantum setting is that the messages
exchanged between the players in different rounds in general do not exist at the same time
due to the no-cloning theorem~\cite{WottersZ:1982, Dieks:1982}. In spite of this, Jain, Radhakrishnan and Sen~\cite{Jain:2008} defined an information theoretic notion of {\em privacy loss} and presented several elegant compression schemes for quantum protocols. The same set of authors
~\cite{Jain:2003:LBB:946243.946331} also proposed
a different measure called {\em information loss} which
extended the work~\cite{Bar-Yossef2002a}, lower bounding the communication complexity of the \textsf{Set-Disjointness} function, to the
quantum setting. Recently, Touchette~\cite{Touchette:2015} has
extended (internal) information complexity to the quantum
setting by defining quantum information complexity (QIC),
inspired by the {\em quantum state-redistribution protocols}~\cite{DevetakY:2008,YardD:2009}. QIC has been shown to satisfy many of the natural properties possessed by IC, and in particular, it is equal to amortized quantum communication complexity. Meanwhile, Touchette has also shown a direct sum
result for {\em bounded-round quantum communication complexity}. To add to these developments, Braverman, Garg, Ko, Mao and Touchette~\cite{BravermanGKMT:2015} have used QIC in a crucial way to give a nearly tight bound on the bound-round quantum communication complexity of \textsf{Set-Disjointness}.
More recently, Nayak and Touchette~\cite{NayakT:2016} used QIC to extend the work of Jain and Nayak~\cite{JainN:2014}, using Augmented Index to lower bound the space complexity of streaming algorithms for \textsc{dyck}(2).

We study the gap between quantum communication complexity (QCC)
and IC. It is known that $\text{QCC}\br{f, 1/3}\leq2^{\O\br{\text{QIC}\br{f, 1/3}}}\leq2^{\O\br{\text{IC}\br{f, 1/3}}}$~\cite{BravermanGKMT:2015} for
any Boolean function $f$, where $\text{QCC}\br{f, 1/3}, \text{IC}\br{f, 1/3}$ and $\text{QIC}\br{f, 1/3}$ represent the minimum QCC, the minimum IC and the minimum IC of a protocol that computes $f$ with error at most $1/3$, respectively.
However, in contrast to the classical analog of this result, their proof does not proceed via a direct compression argument and much remains to be added in our understanding of interactive quantum compression.

\subsection{Results and Contributions}

In this paper we show that
there exists a Boolean function with an exponential gap between its QCC and IC.
This gap is as large as possible~\cite{Braverman2012}.

\begin{theorem}\label{thm:main}
	There exists a (family of) Boolean function $f$ and a distribution $\mu$ on its input such that $\text{QCC}\br{f, \mu, 1/3}\geq2^{\Omega\br{\text{IC}\br{f,\mu, 1/3}}}\geq2^{\Omega\br{\text{QIC}\br{f,\mu, 1/3}}}$. 
\end{theorem}
Combining with the fact that QIC is equal to the amortized quantum communication complexity,
this shows that a tight direct sum result for distributional QCC is not possible. In fact, our results show that for the task we consider, the amortized  \emph{classical} communication is exponentially smaller than the \emph{quantum} communication complexity.
Notice that for the \textsf{Vector-in-Subspace} Problem, Kerenidis, Laplante, Lerays, Roland and Xiao~\cite{Kerenidis2012} proved that its \emph{quantum} communication complexity is exponentially smaller than its amortized  \emph{classical} communication. Our results thus imply that these two notions, QCC and IC, are incomparable.

In~\cite{GanorKR:2014,GanorKR:2015}, Ganor et.al. introduced the \textsf{Bursting-Noise} function and proved that the RCC of this function is exponentially larger than its IC. To this end, they introduced a new lower bound method for RCC, namely the {\em relative discrepancy bound}, and showed that the relative discrepancy bound of the bursting noise function is exponentially larger than the IC. An immediate question, which would directly imply Theorem~\ref{thm:main}, is whether the relative discrepancy bound is also a lower bound on QCC or they are polynomially equivalent. The answer is negative.
In ~\cite{Regev:2011}, Klartag and Regev essentially showed that the relative discrepancy bound of \textsf{Vector-in-Subspace} problem is $\Omega\br{n^{1/3}}$, while the QCC is $\O\br{\log n}$. Later, Rao and Sinha~\cite{RaoS:2015} simplified Ganor et.al's result by defining a similar but relatively simpler  function called \textsf{ Symmetric $k$-ary Pointer Jumping} function, a symmetrized variant of the \textsf{Iterated Index} function~\cite{Klauck:2001:IQC:380752.380786}. They introduced and used the {\em fooling distribution method} to prove the lower bound on the RCC of this function. However, in the same paper, they also showed that fooling distribution method subsumes the relative discrepancy bound, so that we cannot directly rely on their fooling distribution method to prove our desired separation. Currently, other than QIC, the strongest method to prove QCC lower bounds is $\gamma_2$/{\em generalized discrepancy}~\cite{Klauck:2007:LBQ:1328722.1328729,Sherstov:2008:PMM:1374376.1374392}. However, at least in the prior-free setting, the generalized discrepancy is known to be upper bounded by QIC due to~\cite{BravermanGKMT:2015}. Moreover, in the distributional setting, the generalized discrepancy is known to lower bound IC~\cite{Kerenidis2012}, which we know is low for the task we consider.
In particular, our result imply that for some specific functions, like the one we consider here, the generalized discrepancy bound can be exponentially smaller than the QCC.
Hence, to prove Theorem~\ref{thm:main}, we need new techniques to prove the lower bound on QCC.

The function we use to exhibit the exponential separation is the \textsf{ Symmetric $k$-ary Pointer Jumping} function, the same function used by Rao and Sinha~\cite{RaoS:2015} to show the exponential gap between RCC and IC. To reach our goal of showing that QCC is
also large, we adopt the same framework as developed in~\cite{RaoS:2015}, and essentially show that for their task, the fooling distribution they defined is also a \emph{quantum fooling distribution}. However, the proof technique is significantly different from theirs. As explained above, a distribution fooling classical protocols with low communication does not necessary fools quantum protocols with low communication. Moreover, the proof in~\cite{RaoS:2015} heavily relies on two ideas that have no clear quantum counterparts: first,  that a protocol with low communication induces large monochromatic rectangles, and, second,  that given a protocol with input $XY$ drawn from a product distribution and a transcript $M$, $X-M-Y$ forms a Markov chain.

In order to avoid these obstacles, our proof is based on the
 {\em round elimination}
technique \cite{MiltersenNSW:1998, Klauck:2001:IQC:380752.380786,Jain:2003:LBB:946243.946331}.
Even though we handle various technical difficulties surrounding quantum messages, we believe that, conceptually, the high-level outline of our proof, as described in section~\ref{sec:highlevel}, is the simplest among aforementioned exponential separation results, simplifying further the ideas developed in~\cite{RaoS:2015}.

In particular, it is a simple consequence of our proof techniques that the \textsf{Greater-Than} function on $n$ bits satisfies a communication trade-off similar to that of the \textsf{Index} function
\begin{theorem}\label{th:gtinformal}
In any (quantum) protocol computing \textsf{Greater-Than} on $n$ bits with error $1/3$,  if Bob communicates $b$ bits to Alice, then Alice must communicate $\frac{n}{2^{O (b)}}$ bits to Bob.
\end{theorem}
We provide a simple matching upper bound. To the best of our knowledge, this trade-off was not known before, even for classical communication~\cite{BravermanWeinstein2011,Viola:2013,RSinha:2015}.This trade-off is the same as the one of \textsf{Index} function~\cite{MiltersenNSW:1998,JainRS09}, where Alice and Bob are given $x\in\set{0,1}^n$ and $i\in[n]$, respectively,  and \textsf{Index}$\br{x,i}\defeq x_i$. In contrast to \textsf{Index} for which the upper bound can be achieved with only 2-messages (if Bob sends the first message), the protocol we give here to achieve the trade-off requires $\Omega\br{b}$ rounds of interaction if Bob sends fewer bits to Alice than Alice sends to Bob.
Interaction is necessary here, since for any constant number of rounds $r$, the $r$-round communication complexity of \textsf{Greater-Than} on $n$ bits is $\Omega (n^{1/r})$~\cite{MiltersenNSW:1998}.

We point out that the first communication task to be presented as a candidate separating information complexity from communication complexity~\cite{Braverman:2013} was motivated by the \textsf{Greater-Than} function, and all tasks achieving such a separation have a hard distribution bearing some resemblance to the hard distribution for \textsf{Greater-Than}. We build on~\cite{RSinha:2015}, who gave a simple proof of the optimal symmetric $\Omega (\log n)$ lower bound, and apply our strengthening of a lemma, variants of which have appeared in all previous works on exponential separation between IC and RCC.

\section{The Function: Symmetric $k$-ary Pointer Jumping}

To exhibit an exponential separation between QCC and IC, we consider the \textsf{Symmetric $k$-ary Pointer Jumping} function introduced in \cite{RaoS:2015}, which in turn is based upon the ideas introduced in \cite{Braverman:2013, GanorKR:2014,GanorKR:2015}; see Figure~\ref{fig:pointerjumping}. 

We work with the set $[k]=\{0,1,\ldots k-1\}$, endowed with addition (modulo $k$), and strings of elements from this set. For any integer $j$, the set of all strings of length less than $j$ will be represented by $[k]^{<j}$. Another parameter characterizing this function is $n$.  The functions $x,y : [k]^{<n}\rightarrow [k]$ map strings of length less than $n$ to elements of $[k]$. The functions $f,g : [k]^n\rightarrow \{0,1\}$ map strings of length $n$ to binary values $\{0,1\}$. Given an integer $j$ and a string $z$ with $|z| \geq j$ , let $z_{\leq j}$ represent the string formed by taking the first $j$ characters of $z$. Similarly, we define $z_j$ as the $j$-th character of $z$.
We use similar notation for the functions $x, y$, with $x_{\leq j}$ the restriction of $x$ to strings $z$ satisfying $|z| \leq j$, etc.

\begin{figure}
\begin{overpic}[width=0.9\textwidth]{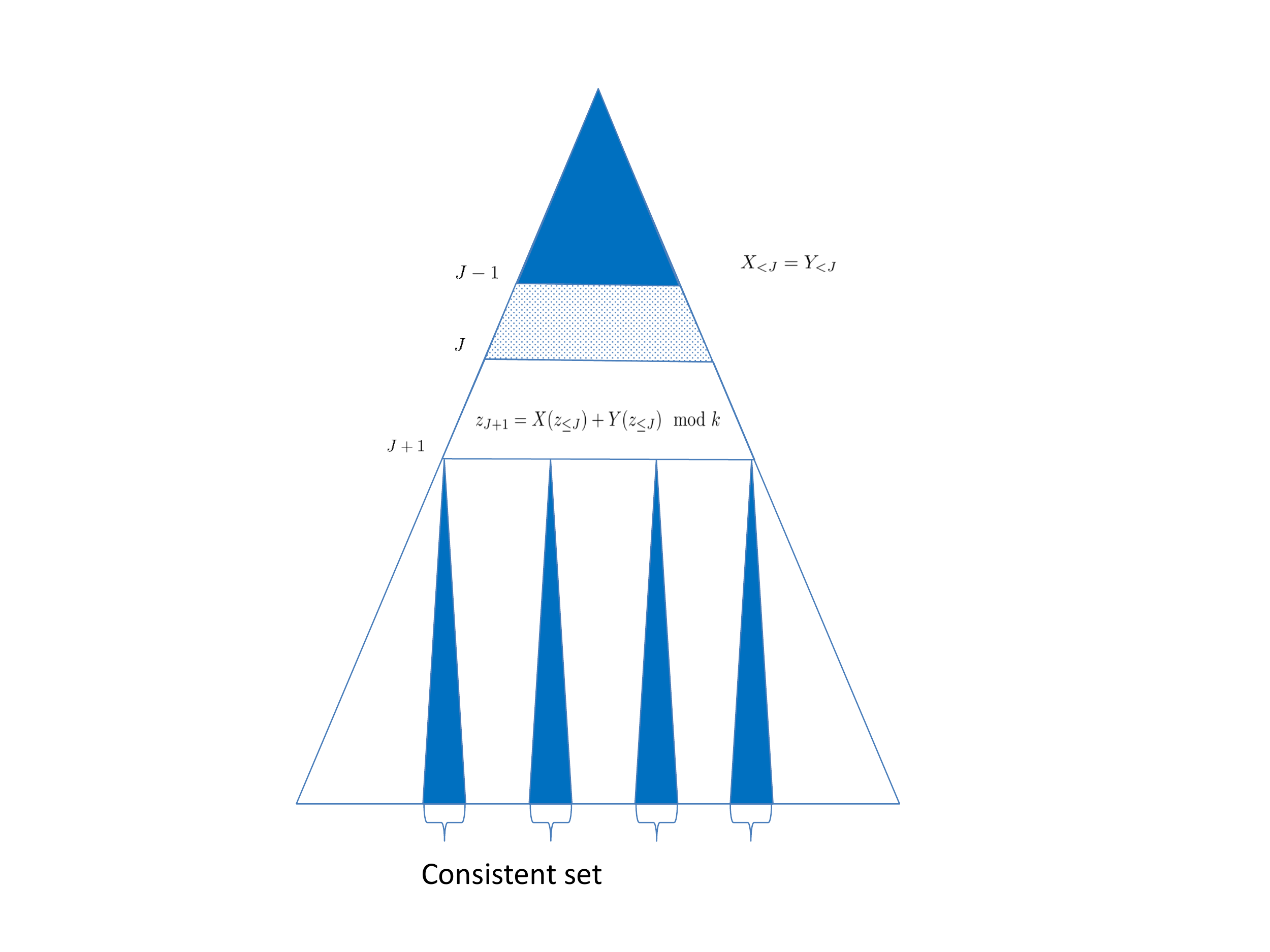}
\end{overpic}
  \caption{Depiction of the $k$-ary pointer jumping function. $X$ and $Y$ are defined for all internal nodes in a complete $k$-ary tree of depth $n$, and $F$ and $G$ are defined for all leaves. Given an hidden layer $J$, it holds that $X_{<J} = Y_{<J}$, and the set of consistent strings is defined through $X_J + Y_J$~mod~$k$. Under $\mu_b$, $X_{>J} = Y_{>J}$ for all consistent internal nodes, and $F \oplus G = b$ for all consistent leaves.}
  \label{fig:pointerjumping}
\end{figure}

For an integer $j<n$ and functions $x,y$, we say that a string $z$ is \textit{consistent} with $x,y,j$ if $|z|>j$ and it holds that
$x(z_{\leq j}) + y(z_{\leq j}) = z_{j+1} \quad \text{mod } k$.
We follow~\cite{RaoS:2015} and define a \emph{ quantum fooling distribution} $p$ from which  we derive a \emph{hard distribution} $\mu$ by further conditioning $p$ on an event $\E$.
We later show that low communication protocols cannot distinguish between $0$-inputs to the hard distribution and inputs to the fooling distribution, and similarly for $1$-inputs.

\begin{definition}
\label{def:mu}
\textbf{Fooling Distribution $p(x,y,f,g,j)$:} Let $J$ be a random variable taking value uniformly at random in $\{0,1\ldots n-1\}$. We define $p(x,y,f,g,j) \defeq \Pr_J(j)\cdot p(x,y,f,g|j)$, where the conditional distribution $p(x,y,f,g|j)$ is defined as follows: $x,y,f,g$ are chosen uniformly at random, subject to the constraint that for all $z \in [k]^{<j}$, $x(z)=y(z)$.
\end{definition}

\begin{definition}
\label{def:p}
\textbf{Hard Distribution $\mu(x,y,f,g,j)$:} Let $\E_0$ be the event that for every $x,y,f,g,j$ and every $z$ consistent with this choice of $x,y,j$, $x(z)=y(z)$ (when $|z|<n$) and $f(z)=g(z)$ (when $|z|=n$). Let $\E_1$ be the event that for every $x,y,f,g,j$ and every $z$ consistent with this choice of $x,y,j$, $x(z)=y(z)$ (when $|z|<n$) and $f(z)\neq g(z)$ (when $|z|=n$). Let $\E \defeq \E_0 \vee \E_1$, then $\mu(x,y,f,g,j) \defeq p(x,y,f,g,j| \E)$. We further denote $\mu_0 = \mu | \E_0 = p | \E_0$, and $\mu_1 = \mu | \E_1 = p | \E_1$, so that $\mu = \frac{1}{2} \mu_0 + \frac{1}{2} \mu_1$.
\end{definition}

This allows us to define the inputs to Alice and Bob and the required task.
\begin{definition}
\label{def:commtask}
\textbf{The Communication Task.}
\begin{itemize}
\item A referee draws $x,y,f,g,j$ from the distribution $\mu(x,y,f,g,j)$. Alice is given input $(x,f)$, and Bob input $(y,g)$. The index $j$ is kept hidden from both parties.
\item Let $\hat{z}\in [k]^n$ be the unique string that satisfies, for all $r>0$ (and $r<n$),  $x(\hat{z}_{\leq r}) + y(\hat{z}_{\leq r}) = \hat{z}_{r+1}$, and $x(\epsilon) + y(\epsilon) = \hat{z}_1$ for $\epsilon$  the empty string. Alice and Bob must output $f(\hat{z})+g(\hat{z})~\text{mod } 2$.
\end{itemize}
\end{definition}

An important property of the distribution $\mu(x,y,f,g,j)$ is that the output $f(\hat{z})+g(\hat{z})~\text{mod } 2$ is the same on all consistent strings, simply because $f(z)=g(z)$ (or $f(z)\neq g(z)$) on all consistent strings $z$, and the unique string $\hat{z}$ on which $f(\hat{z})+g(\hat{z}) \quad \text{mod } 2$ must be evaluated is also a consistent string. Thus, we define $S$ to be the set of all consistent strings for a given tuple $x,y,j$. This allows us to extend the definition of distributions $p$ and $\mu$ to include $S$, as $p(x,y,f,g,s,j)$ and $\mu(x,y,f,g,s,j)$.

The proof of our main theorem, Theorem~\ref{thm:main}, follows from the following two theorems.

\begin{theorem}
\label{theo:qicupperbound}
There exists a quantum protocol that accomplishes the communication task from Definition~\ref{def:commtask} with error $\varepsilon  \leq \frac{1}{\log n}$ and  with QIC upper bounded by IC, which in turn is upper bounded by $\O(\log(k\log n)2^{\frac{2\log n}{k}})$.
\end{theorem}

\begin{theorem}
\label{theo:qcclowerbound}
Any protocol which accomplishes the communication task from Definition~\ref{def:commtask} with constant error $\varepsilon \in (0, \frac{1}{2})$ requires a quantum communication cost lower bounded by $\min\set{\Omega\br{k^{1/5}},\Omega (\log n)}$.
\end{theorem}
If we choose $k=\log n$, then the IC is $\O\br{\log k}$ while the QCC is $\Omega\br{k^{1/5} }$.

Our technical contributions go into proving the lower bound on QCC stated in Theorem~\ref{theo:qcclowerbound}. The upper bound of QIC in Theorem~\ref{theo:qicupperbound} follows by combining the two theorems below, proven in~\cite{RaoS:2015} and~\cite{LauriereT:2016}, respectively.

\begin{theorem}~\cite{RaoS:2015}
There exists a classical protocol that accomplishes the communication task from Definition~\ref{def:commtask} with constant error $\varepsilon > 0$ and  with IC upper bounded by $O(\log(k\log n)2^{\frac{2\log n}{k}})$.
\end{theorem}

\begin{theorem}~\cite{LauriereT:2016}
For any classical protocol $\Pi$, there exists a quantum protocol $\Pi^\prime$ exactly simulating the input-output behavior of $\Pi$ while maintaining the same communication pattern as (the padded version of) $\Pi$, and also satisfying $QIC (\Pi^\prime, \mu) = IC (\Pi, \mu)$ for all $\mu$.
\end{theorem}

In~\cite{LauriereT:2016}, the bulk of the effort for showing the theorem about the
quantum simulation of classical protocols
goes into arguing how to quantumly simulate private randomness without
affecting the information cost. Note that we could alternatively use the
fact that IC is equal to amortized communication complexity to argue
that the IC is also an upper bound on the
QIC for any communication task in the distributional setting: $\text{QIC} (f, \mu, \epsilon) = \text{AQCC} (f, \mu, \epsilon) \leq \text{ACC} (f, \mu, \epsilon) = \text{IC} (f, \mu, \epsilon)$.


\section{High-Level Proof Sketch for the Communication Lower Bound}\label{sec:highlevel}

In this section, we give a high-level proof sketch of Theorem \ref{theo:qcclowerbound}.
We also formally state the main technical lemmata that go into the proof.
Formal proofs are given in Section~\ref{sec:proofs}.
Our strategy for proving the lower bound is divided into two main steps.

\begin{itemize}
\item We first consider the fooling distribution $p (x,y,f,g,j)$ and show that in any quantum protocol $\Pi$ with small communication, the state of the registers with Bob is almost independent of $X_S F_S$, conditioned on $x_{\leq j} y_{\leq j} j$, and similarly the state of the registers with Alice is almost independent of $Y_S G_S$, conditioned on $x_{\leq j} y_{\leq j} j$.
For this, we argue by performing two different reductions to one-round protocols.
\item Using the observation that, conditioned on $x_{\leq j}y_{\leq j}j$, $p (x, y, f, g, j)$ and $\mu(x,y,f,g,j)$ have the same marginals on $(x, f)$,  and also the same marginals on $(y, g)$, we show that the `approximate independence' concluded above for $p (x, y, f, g, j)$ implies that the final state on Alice's or Bob's registers is approximately the same for inputs according to either of $\mu_0\br{x,y,f,b}\defeq p (x,y,f,g|\E_0)$, $\mu_1\br{x,y,f,g}\defeq p (x,y,f,g|\E_1)$ or $p (x,y,f,g)$.
For this, we argue by performing a round-by-round elimination.
\end{itemize}

\bigskip
\begin{figure}[ht]
\centering
\begin{tikzpicture}[xscale=1.2,yscale=1.4]
\draw (-3.4,-0.6) rectangle (5,4.1);
\node at (-1.0,3.8) {Quantum Fooling Distribution $p$};
\draw (-2.9,2.8) rectangle (0.1,3.5);
\node at (-1.4,3.3) {Low communication};
\node at (-1.4,3.0) {multiround protocol};
\draw [->] (-1.4, 2.8) -- (-1.4,2.4);
\draw (-2.9,1.4) rectangle (0.1,2.4);
\node at (-1.4,2.2) {Convert to one-way};
\node at (-1.4,1.9) {protocol with abort};
\node at (-1.4,1.6) {+ Lemma \ref{lem:onerddirectsum} };
\draw [->] (-1.4, 1.4) -- (-1.4,1.0);
\draw (-3.1,-0.5) rectangle (0.3,1.0);
\node at (-1.4,0.8) {Lemma \ref{lem:multirddirectsum} $\implies$};
\node at (-1.4,0.5) {$ \condmutinf{X_J }{BY_{>J}G}{J X_{<J}},$};
\node at (-1.4,0.1) {$ \condmutinf{Y_J }{AX_{>J}F}{J Y_{<J}}$};
\node at (-1.4,-0.2) {are small};

\draw [->] (0.3,0.25) -- (1.7,0.25) -- (1.7,2.8);

\draw [->] (0.1,3.1) -- (1.4,3.1);
\draw (1.4,2.8) rectangle (4.8,3.5);
\node at (3.1,3.3) {Convert to one-way};
\node at (3.1,3.0) {protocol: Lemma \ref{lem:coherentJRS}};
\draw [->] (3.5,2.8) -- (3.5,2.4);
\draw (2.2,1.6) rectangle (4.8,2.4);
\node at (3.5,2.15) {Quantum Shearer};
\node at (3.5,1.85) {Lemma \ref{lem:shearer}};
\draw [->] (3.5,1.6) -- (3.5,1.2);
\draw (1.9,-0.5) rectangle (4.8, 1.2);
\node at (3.4,1) {Low correlation of};
\node at (3.4,0.7) {$X_SF_S$ with $BY_{>J}G$};
\node at (3.4,0.4) {$Y_SG_S$ with $AX_{>J}F$};
\node at (3.4,0.1) {conditioned on};
\node at (3.4,-0.3) {$X_{\leq J}Y_JJ$};

\draw [->, ultra thick] (4.8,0.2) -- (7,0.2);
\node at (6,0.35) {Lemma \ref{lem:rdelim}};
\draw (7,-0.4) rectangle (9.2,1);
\node at (8.1,0.7) {Distributions};
\node at (8.1,0.3) {$\mu_0,\mu_1$ give};
\node at (8.1,-0.1) {similar outputs};

\end{tikzpicture}
\caption{Structure of our proof. We have ignored purification of input registers for simplicity of presentation.}
 \label{fig:proofstrategy}
\end{figure}
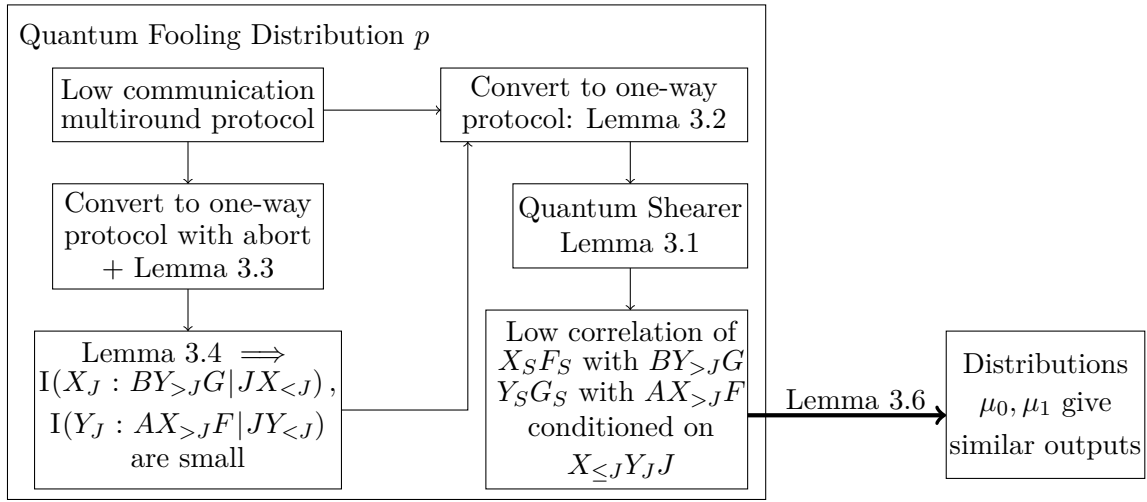
\bigskip

A sketch of our proof strategy appears in Figure \ref{fig:proofstrategy}. In more details, let us first consider the simpler case of a single-message protocol from Alice to Bob, under distribution $p$, with some fixed value of $y_{\leq j}j$. As discussed above, we show that the output under $p$ and the output under $\mu_0$ are close, given that the message is short. A similar argument holds for $\mu_1$, leading to a contradiction. Denote by $M_1$ the register holding the first message (and possibly some pre-shared entanglement). Notice that for a single message, since the marginal on $(x, f)$ is the same in $p$ and $\mu_0$, the state on registers $XFM_1$ is also the same under these two distributions. But the correlations with Bob's input $(y, g)$ are different: since $XF$ is independent of $YG$ under $p$ (conditioned on the fixed value of $y_{\leq j}j$),  $M_1$ is also independent of $YG$; whereas under $\mu_0$ (and similarly $\mu_1$), $X_SF_S = Y_SG_S$ which means that $M_1$ is highly correlated with $Y_SG_S$ (more precisely $Y_SG_S M_1 = X_SF_SM_1$). Notice that on restricting to the complement of $S$, $Y_{>J} G$ is independent of $ X_S F_S M_1$ and distributed in the same way under both $p$ and $\mu_0$.
Now, the distance between the final output under $p$ and under $\mu_0$ can be upper bounded, using monotonicity, by the distance between $Y_S G_S \otimes M_1$ (under $p$) and $Y_S G_S M_1$ (under $\mu_0$). By the above argument, this is same as the distance between $X_S F_S \otimes M_1$ (under $p$) and $X_S F_S M_1$ under $\mu_0$ (which is distributed as $X_SF_S M_1$ under $p$).
 This is in turn upper bounded by the mutual information between $X_SF_S$ and $M_1$ under the distribution $p$. To complete the argument, we use the following lemma, which can be thought of as a quantum version of Shearer's Lemma~\cite{CHUNGGFS:1986, Radhakrishnan:2003} for mutual information. 

\begin{lemma}\label{lem:shearer}
	Consider registers $U_1,U_2,\ldots U_m, V$ and define $U\defeq U_1U_2,\ldots U_m$. Consider a quantum state $\Psi_{UV}$ such that $\Psi_{U_1,U_2,\ldots U_m} = \Psi_{U_1}\otimes \Psi_{U_2}\otimes\ldots \otimes \Psi_{U_m}$. Let $S=\set{i_1,\ldots,i_{|S|}}\subseteq [m]$ be a random set independent of $\Psi_{UV}$ satisfying $\prob{i\in S}\leq\frac{1}{k}$ for all $i$ and $U_S\defeq U_{i_1}U_{i_2}\ldots U_{i_{|S|}}$. Then it holds that
	\[\condmutinf{U_S}{V}{ S}_{\Psi}\leq\frac{\mutinf{U}{V}_{\Psi}}{k},\]
\end{lemma}

Now, to extend the above argument to multi-round protocols, we want to ensure that even if Alice knows some information about $S$, the argument still goes through, as long as her information about $S$ is small. We do so by specially crafting an input to the protocol and then reducing it to an essentially equivalent one-round protocol. For this, we use an asymmetric round-compression argument from~\cite{Jain:2005:PEM:1068502.1068658} to generate the state in each round of the protocol, up to a small error, by a one-way protocol with communication cost close to that in the original protocol. We also require a similar argument on Bob's side.
Formally, we prove the following result, with some extra care needed since we wish, for technical reasons, to maintain correlations with the reference registers.

\begin{lemma}\label{lem:coherentJRS}

Consider a quantum state $\ket{\Psi}=\sum_{xy}\sqrt{\mu\br{x, y}}\ket{xxyy}_{R_XXR_YY}\otimes\ket{\psi^{xy}}_{AB}$ satisfying $\mutinf{Y}{R_XXA}_{\Psi}\leq\epsilon$, where $\mu=\mu_X\otimes\mu_Y$ is a product distribution, and the register $X$ and register $Y$ are held by Alice and Bob, respectively. Given $\delta> 0$, there exists a one-way quantum protocol where Alice sends $\O\br{\br{\mutinf{X}{YR_YB}_{\Psi}+1}/\delta^2}$ qubits to Bob. Let $\tilde{\Psi}$ be the global state in the end of this protocol. It holds that \footnote{$h(\cdot,\cdot)$ denotes the {\em Hellinger distance} which will be defined in section~\ref{sec:prelim}.}
	\[h^2\br{\tilde{\Psi}_{XABYR_Y},\Psi_{XABYR_Y}}\leq 4 \delta^2 + 6\epsilon,\]

\end{lemma}

To prove that the information about $S$ is small, first notice that for fixed $x_{\leq j} j$, $S$ is determined by $y_j$, and vice-versa. Hence, we wish to bound the amount of  information about $Y_j$ that Alice has in any round, conditioned on some fixed values of $x_{\leq j} j$. In all previous works~\cite{GanorKR:2015b,GanorKR:2015,RaoS:2015} on exponential separation between information and communication, the proof relied on a statement of the form ``the information Alice has about the $j$-th part of Bob's input is upper bounded by $\frac{2^{O (\ell)}}{n}$''. This holds even when conditioning on some $j$ playing a role similar to the hidden index $j$ here, and also on some part of Alice's input corresponding to $j$.  $\ell$ is the total number of bits of  communication  in the protocol, and $n$ is the number of parts of Alice's input (usually related to the depth of some underlying communication tree), of size exponentially larger than the desired communication bound. This is usually proved via involved information-theoretic arguments that make use of the rectangular nature of classical protocols,
hence such proof cannot be generalized to the quantum setting at all. We give a very simple two-step argument to achieve similar bounds. First, we once again use a reduction to a one-way protocol. Second, for such one-way protocols, we can use a simple direct sum argument and avoid the exponential blow-up. Formally, we have the following lemma for one-way protocols, variants of which have appeared in~\cite{Klauck:2001:IQC:380752.380786, Sen:2001}.

\begin{lemma}\label{lem:onerddirectsum}

Let $\Pi$ be a quantum one-way protocol with correlated inputs $XY$, in which Alice sends $\ell$ qubits to Bob. Let $X = X_1 \cdots X_n$, and for a uniformly random index $J  \in_R [n]$, decompose $Y = Y_1^J Y_2^J$ such that $Y_2^J$ is a function of $X_{<J}$ and $\br{X_{\geq J} Y_1^J | J X_{<J}=jx_{<j}} = \br{X_{\geq J} \otimes Y_1^J | J X_{<J}=jx_{<j}}$ for any $jx_{<j}$, that is, conditioned on $J$ and $X_{<J}$, $X_{\geq J}$  and $Y_1^J$ are independent. Let $\rho_{XR_XYR_YABC}$ be the global state in the end of the protocol, where $A$ is the register with Alice; $C$ is the register of the message Alice sends to Bob; $B$ is the register with Bob before receiving the message and $R_XR_Y$ are the canonical purification of the input $XY$. Then it holds that
\begin{align}
\condmutinf{X_J}{C B Y_1^J R_{Y_1^J}}{ J X_{<J}}_{\rho} \leq \frac{2\ell}{n}.
\end{align}
\end{lemma}

Second, to extend the lemma to multiple-round protocols, we still have ``enough room'' to perform a one-way simulation of any interactive protocol, with at most an exponential blow-up in the communication and still achieve similar bounds as in the classical setting. Formally, we prove the following result
 by appealing to both compression arguments and to the notion of protocols with abort~\cite{Kerenidis2012,LaplanteLR:2012},
with some extra care needed since we again wish, for technical reasons, to maintain correlations with the reference registers.

\begin{lemma}\label{lem:multirddirectsum}

Let $\Pi$ be a quantum protocol with correlated input $XY$. Let $X = X_1 \cdots X_n$, and for a uniformly random index $J  \in_R [n]$, decompose $Y = Y_1^J Y_2^J$ such that $Y_2^J$ is a function of $X_{<J}$ and $\br{X_{\geq J} Y_1^J | J X_{<J}} = \br{X_{\geq J} \otimes Y_1^J | J X_{<J}}$, that is, conditional on $JX_{<J}$, $X_{\geq J}$  and $Y_1^J$ are independent. Then, for any $r$, it holds that
\begin{align}
 \condmutinf{X_J }{C_rB_r Y_1^J R_{Y_1^J} }{ J X_{<J}} \leq   \frac{\ell_{A,r}2^{{2\ell_{B,r}+2}}}{n},
\end{align}
where $\ell_{A,r}$ and $\ell_{B,r}$ are the number of qubits Alice and Bob send in the first $r$ rounds, respectively.

\end{lemma}

Finally, in order to go from the distribution $p$ to the distribution $\mu_0$, we have the following \emph{distributional cut-and-paste} lemma. Intuitively, it states the following. Assume that in each round and on a product input distribution, the local states are almost independent of the other party's input. Then, up to local isometries, the overall state stays independent of the joint input. Importantly, this holds even after conditioning the input distribution on an arbitrary joint event. Hence, if the input is replaced by another one with the same marginal distributions on both sides, then the marginals of the global state in the final round on both sides are almost unchanged. Note that $p$ and  $\mu_0$ have the same marginal distributions on the both sides and $p$ is a product distribution conditioned on $x_{\leq j}y_{\leq j}j$. Thus the following lemma enables us to show that neither Alice nor Bob is able to distinguish $p$ from $\mu_0$ and equivalently $p$ from $\mu_1$. The lemma could be interesting on its own and we believe it should have other applications in quantum communication complexity. The proof is inspired from quantum versions of the cut-and-paste lemma~\cite{Jain:2003:LBB:946243.946331, JainN:2014, NayakT:2016}, with extra care needed to go from one distribution to the other. Let us set some notation before stating the lemma.

\begin{definition}
	\label{def:rdelim}
	
	Consider a protocol $\Pi$, and states $\ket{\rho}_{X R_X Y R_Y} =\ket{\rho}_{X R_X} \otimes \ket{\rho}_{Y R_Y}$ and $\ket{\sigma}_{X Y R_X R_Y}$
	such that $\sigma_X = \rho_X$, $\sigma_Y = \rho_Y$,
	and $\rho_{XY} = \rho_X \otimes \rho_Y$ and $\sigma_{XY}$ are classical input distributions for $\Pi$ with canonical purifications $\ket{\rho}_{X R_X Y R_Y}$ and $\ket{\sigma}_{X Y R_X R_Y}$, respectively.
	We denote by $\ket{\rho^i}_{X   R_X Y R_Y A_i B_i C_i}$ and $\ket{\sigma^i}_{X   R_X Y R_Y A_i B_i C_i}$
	the state in round $i$ when $\Pi$ is run on input distributions $\rho_{XY}$
	and $\sigma_{XY}$, respectively.
	For any register $L$, we use $\tilde{L}$ to represent a new register with the same dimension as $L$.
	For $i>0$ odd, let
	\begin{align}
	\epsilon_i & \defeq h (\rho^i_{R_X Y R_Y B_i C_i}~,~\rho^i_{R_X} \otimes \rho^i_{Y R_Y B_i C_i}),
	\end{align}
	and for $i>0$ even,
	\begin{align}
	\epsilon_i & \defeq h (\rho^i_{R_Y X R_X A_i C_i}, \rho^i_{R_Y} \otimes \rho^i_{X R_X A_i C_i}).
	\end{align}
	For $i=0$, let $C_0 = 1$ be a trivial register, let $\epsilon_0 = 0$
	and let
	\begin{align}
	V^0 & \defeq I_{Y} \otimes I_{B_0 \rightarrow \tilde{B}_0} \otimes V^Y_{1 \rightarrow \tilde{Y}_0 \tilde{R}_{Y_0}},
	\end{align}
	in which $V^Y_{1 \rightarrow \tilde{Y}_0 \tilde{R}_{Y_0}}$ creates $\ket{\rho^Y}_{\tilde{Y}_0 \tilde{R}_{Y_0}}$ from nothing.

	Also let, for odd $i > 0$,	$V^i  = V_{X A_i \rightarrow X \tilde{A}_i \tilde{X}_i \tilde{R}_{X_i}}^i, $
	satisfying
	\begin{align}
	\epsilon_i & = h (~V^i (\rho^i_{X R_X Y R_Y A_i B_i C_i})~,~
	\rho_{X R_X} \otimes \rho^i_{ \tilde{X}_i \tilde{R}_{X_i} Y R_Y \tilde{A}_i B_i C_i})~, \\
	\end{align}

	(note that $B_i = B_{i-1}$ for odd $i > 0$, and $A_i = A_{i-1}$ for even $i >0$) and for $i>0$ even, $V^i  = V^i_{Y B_i \rightarrow Y \tilde{B}_i \tilde{Y}_i \tilde{R}_{Y_i}}$ satisfying
	
	\begin{align}
	\epsilon_i	& = h (~V^i (\rho^i_{X R_ X Y R_Y A_i B_i C_i})~,~
	\rho_{Y R_Y} \otimes \rho^i_{X R_X  \tilde{Y}_i \tilde{R}_{Y_i} A_i \tilde{B}_i C_i})~. \\
	\end{align}

	The existence of $V^i$'s is guaranteed by Fact~\ref{fac:Uhlmann}.

\end{definition}

\begin{lemma}
	\label{lem:rdelim}
	
With the notation from Definition~\ref{def:rdelim}, let, for odd $i > 0$,

\begin{align}
	\gamma_i & = h (~V^i V^{i-1} (\rho^i_{X R_X Y R_Y A_i B_i C_i}) ~,~
	\rho_{X R_X Y R_Y} \otimes  \rho^i_{ \tilde{X}_i \tilde{R}_{X_i} \tilde{Y}_{i-1} \tilde{R}_{Y_{i-1}} \tilde{A}_i \tilde{B}_i C_i})~,
	\end{align}
	and
	\begin{align}
	\delta_i & = h (~V^i V^{i-1} (\sigma^i_{X R_X Y R_Y A_i B_i C_i})~,~
	\sigma_{X Y R_X R_Y}  \otimes  \rho^i_{ \tilde{X}_i \tilde{R}_{X_i} \tilde{Y}_{i-1} \tilde{R}_{Y_{i-1}} \tilde{A}_i \tilde{B}_i C_i}),
	\end{align}

and for $i> 0$ even, let

\begin{align}
	\gamma_i & = h (~V^i V^{i-1} (\rho^i_{X R_X Y R_Y A_i B_i C_i})~,~
	\rho_{X R_X Y R_Y} \otimes  \rho^i_{ \tilde{X}_{i-1} \tilde{R}_{X_{i-1}} \tilde{Y}_{i} \tilde{R}_{Y_{i}} \tilde{A}_i \tilde{B}_i C_i})~,
	\end{align}
	and
	\begin{align}
	\delta_i & = h (~V^i V^{i-1} (\sigma^i_{X R_X Y R_Y A_i B_i C_i})~,~
	\sigma_{X Y R_X R_Y}  \otimes  \rho^i_{ \tilde{X}_{i-1} \tilde{R}_{X_{i-1}} \tilde{Y}_{i} \tilde{R}_{Y_{i}} \tilde{A}_i \tilde{B}_i C_i}).
	\end{align}
	Then it holds that for $i \geq 1$,
	$$
	\gamma_i \leq \epsilon_i + \epsilon_{i-1} + 2 \sum_{j=1}^{i-2} \epsilon_j  , \quad
	\delta_i \leq \epsilon_i + \epsilon_{i-1} + 2 \sum_{j=1}^{i-2} \epsilon_j.
	$$
	
\end{lemma}

The theorem follows by blending all of these ingredients together,  using a concavity argument, and also optimizing over the number of rounds $t$.

Also, note that the polynomial rather than linear dependence on $k$ is due to the last round-elimination argument, in Lemma~\ref{lem:rdelim}, which works in a round-by-round fashion and from which
a factor of $t$, the number of rounds, comes out and over which we must optimize. The other lemmata do not incur such blow-up, and if we take the corresponding lemmata in the classical setting, we could further use the Markov property of classical protocol run on product distributions along with the specific ``$x = y$'' event, as done in Lemma~5 in~\cite{RaoS:2015} in order to obtain a tight $\Omega (k)$ lower bound. Obtaining tight round elimination arguments in the quantum setting remains an important open question, and another interesting open question is whether one can avoid such a round-by-round argument, and the extra factor of $t$ coming out of it, to complete the proof in the quantum setting as well.


\section{Warm-up: Trade-off for Greater-Than}\label{sec:gtthan}

In this section, we investigate the trade-off between the communication from Alice to Bob and the one from Bob to Alice for \textsf{Greater-Than} function. For $x,y\in\set{0,1}^n$, we define $x\geq y$ if the integer with binary representation $x$ is at least as large as the integer with binary representation $y$. The \textsf{Greater-Than} function is defined as
\[\textsf{Greater-Than}\br{x,y}\defeq\begin{cases}
1~ \text{if $x\geq y$},\\
0~ \text{otherwise}.
\end{cases}.\]

Let us restate Theorem~\ref{th:gtinformal} more formally.

\begin{theorem}
	Given any constant $0<\epsilon<\frac{1}{2}$ and a quantum protocol that computes \textsf{Greater-Than}: $\set{0,1}^n\times\set{0,1}^n\rightarrow\set{0,1}$ with error at most $\epsilon$, if Bob communicates $b$ qubits to Alice, then Alice must communicate at least $\frac{n}{2^{\Omega\br{b + 1}}}$ qubits to Bob. Moreover, this trade-off is tight.
\end{theorem}

\begin{proof}
	By a standard repetition argument, we may assume without loss of generality that $\epsilon$ is a sufficiently small constant ; this can at most increase Alice's and Bob's respective communication by a constant multiplicative factor. Suppose Alice communicates $a \geq 1$ qubits. Then by the proof of Lemma~\ref{lem:multirddirectsum}, there exists a one-way quantum protocol that computes \textsf{Greater-Than} with communication $a\cdot 2^{\O\br{b}}$ and error at most $2\epsilon$. Thus it suffices to show that the quantum one-way communication complexity of \textsf{Greater-Than} is $\Omega\br{n}$. Our proof is close to the one in~\cite{RSinha:2015}, where Ramamoorthy and Sinha provided a tight lower bound on the RCC of \textsf{Greater-Than},  $\Omega\br{\log n}$.  We adopt the hard distribution of the inputs given in~\cite{RSinha:2015} (slightly adapted from \cite{BravermanWeinstein2011,Viola:2013}) and show that the distributional quantum one-way communication complexity of \textsf{Greater-Than} under this distribution is $\Omega\br{n}$. Then we further apply Yao's minimax theorem~\cite{Yao:1979:CQR:800135.804414} to get the desired lower bound.
	
	Let $J\in[\frac{n}{2}]$ be uniformly random. $X,Y\in\set{0,1}^n$ are sampled uniformly conditioned on the event that $X_{<J}=Y_{<J}$, where $X_{<J}\defeq X_1\ldots X_{J-1}$.
	Let $\Pi$ be a quantum one-way protocol that computes \textsf{Greater-Than} with communication at most $c$ and error at most $2\epsilon$. We use $ACB$ to represent the state shared between Alice and Bob after Alice sends the message, where $A$ is the remaining register with Alice; $C$ is the register sent to Bob and $B$ is the register owned by Bob in the beginning of the protocol ($B$ is independent of the inputs). $C$ contains at most $c$ qubits. Consider
	\begin{eqnarray}\label{eqn:1}
		&&\condmutinf{CBY}{X_J}{X_{<J}J}=\expec{j\leftarrow J}{\condmutinf{CB}{X_j}{X_{<j}j}}=\expec{j\leftarrow J}{\condmutinf{CB}{X_j}{X_{<j}}}=\frac{2}{n}\mutinf{CB}{X_{\leq\frac{n}{2}}}\nonumber\\
		&=&\frac{2}{n}\condmutinf{C}{X_{\leq\frac{n}{2}}}{B}\leq \frac{4c}{n};
	\end{eqnarray}
	where the second equality is from the fact that $J$ is independent of $CBX_J$ given $X_{<J}$; the third equality is by the chain rule; the fourth equality is from the fact that $B$ is independent of the inputs; the inequality is from Fact~\ref{lem:boundoncmi}.
	Let $O$ be the output of the protocol. The following claim is proved in~\cite{RSinha:2015}.
	\begin{claim}\cite{RSinha:2015}\label{claim:gt}
		Suppose $n>20$, it holds that
		\begin{equation}\label{eqn:gt1}
			\condmutinf{\textsf{Greater-Than}\br{X,Y}}{O}{X_{<J}Y_{<J}J}\geq1-\O\br{\sqrt{\epsilon}\log\frac{1}{\epsilon}},
		\end{equation}

and 
		\begin{equation}\label{eqn:gt2}
			\condmutinf{\textsf{Greater-Than}\br{X,Y}}{O}{X_{\leq J}Y_{<J}J}<0.84.
		\end{equation}
			\end{claim}
	Hence,
	\begin{eqnarray*}
		&&\condmutinf{\textsf{Greater-Than}\br{X,Y}}{O}{X_{<J}Y_{<J}J}\\
		&\leq&\condmutinf{X_J\textsf{Greater-Than}\br{X,Y}}{O}{X_{<J}Y_{<J}J}\\
		&\leq&\condmutinf{X_J}{O}{X_{<J}Y_{<J}J}+\condmutinf{\textsf{Greater-Than}\br{X,Y}}{O}{X_{\leq J}Y_{<J}J}\\
		&\leq&\condmutinf{X_J}{CBY_{\geq J}}{X_{<J}Y_{<J}J}+\condmutinf{\textsf{Greater-Than}\br{X,Y}}{O}{X_{\leq J}Y_{<J}J}\\
		&\leq&\frac{4c}{n}+0.84;
	\end{eqnarray*}
	where the third inequality is from Fact~\ref{fact:subsystem monotone} and the last inequality is from Eqs.~\eqref{eqn:1} and~\eqref{eqn:gt2} . Combining with Eq.~\eqref{eqn:gt1}, the result follows.
	
	To prove the tightness, let's assume without loss of generality that Alice sends more qubits to Bob than Bob sends to Alice. It is well-known that the RCC of $\textsf{Greater-Than}$ with bounded error is $\O\br{\log n}$ due to Nisan~\cite{Nisan:1993}. Thus it suffices to consider the case that $\frac{n}{2^b}=n^{\Omega\br{1}}$. To achieve such a bound, Alice and Bob first check whether $x=y$ using shared hashing function with $\O\br{1}$ bits. Then, they equally divide the inputs into $2^{\Omega\br{b}}$ intervals of $\frac{n}{2^{\Omega (b)}}$ bits before running the protocol in Fact~\ref{fac:noisybinarysearch} below, in order to find the interval containing the most significant bit for which $x$ and $y$ differ. Alice further sends the part of her input in that interval to Bob, which requires $\frac{n}{2^{\Omega\br{b}}}$ bits, larger than $b$. Hence the total communication from Alice to Bob is $\frac{n}{2^{\Omega\br{b + 1}}}$.
\end{proof}
\begin{fact}\label{fac:noisybinarysearch}\cite{FeigeRPU:1994}
	There exists a randomized public-coin protocol with communication complexity $\O\br{\log k/\epsilon}$ such that on input two strings $x,y\in\X^k$, where $\X$ is a finite set, it outputs the smallest index $i\in[k]$ such that $x_i\neq y_i$ with probability at least $1-\epsilon$, if such $i$ exists.
\end{fact}

\section*{Acknowledgment}
We are grateful to Ashwin Nayak for helpful discussions, Henry Yuen and Makrand Sinha for helpful correspondence. A.A. would like to thank Rahul Jain for related discussions. We thank Robin Kothari for his valuable help in typesetting of equations.

A.A. is supported by the National Research Foundation, Prime Minister's Office, Singapore and the
Ministry of Education, Singapore under the Research Centres of Excellence programme. 
D.T. is supported in part by NSERC, CIFAR, Industry Canada and ARL CDQI program. IQC and PI are supported in part by the Government of Canada and the Province of Ontario.
P.Y. is supported by the Department of Defense. Part of this work was
done when A.A. was visiting Institute for Quantum Computing (IQC), University of Waterloo under Queen Elizabeth Scholarship and P.Y. and N.Y. were postdoctoral fellows at IQC supported by NSERC and CIFAR.


\section{Preliminaries}\label{sec:prelim}
\subsection{Information Theory}
For an integer $n \geq 1$, let $[n]$ represent the set $\{1,2, \ldots, n\}$.
Let $\X$ and $\Y$ be finite sets and $k$ be a natural number.
Let $\X^k$ be the set $\X\times\cdots\times\X$, the Cartesian product of
$\X$, $k$ times. Given $a=a_1,\ldots,a_k$, we write $a_{\leq i}$ to denote $a_1,\ldots,a_i$. We define $a_{<i},a_{\geq i},a_{>i}$ similarly. We write $a_S$ to represent the projection of $a$ to the coordinates specified in the set $S\subseteq[k]$. Let $\mu$ be a probability distribution on $\X$.
Let $\mu(x)$ represent the probability of $x\in\X$ according to $\mu$.
Let $X$ be a random variable distributed according to $\mu$.
We use the same symbol to represent a
random variable and its distribution whenever it is clear from the
context.
The expectation value of function $f$ on $\X$ is defined as
$\expec{x \leftarrow X}{f(x)} \defeq \sum_{x \in \X} \prob{X=x}
\cdot f(x)$, where $x\leftarrow X$ means that $x$ is drawn according to the distribution of $X$.

A quantum state (or just a state) $\rho$ is a positive semi-definite matrix with unit trace. It is called pure if its rank is $1$. For unit vector $\ket{\psi}$,
with slight abuse of notation, we use $\psi$ to represent the state
and also the density matrix  $\ketbra{\psi}$, associated with $\ket{\psi}$.
A classical distribution $\mu$ can be viewed as a diagonal quantum state with entries $\mu(x)$.
For two quantum states $\rho$ and $\sigma$, $\rho\otimes\sigma$ represents the tensor product (Kronecker product) of $\rho$ and $\sigma$.
A quantum super-operator $\E(\cdot)$ is a completely positive and trace preserving (CPTP) linear map from states to states.
Readers can refer to~\cite{CoverT91,NielsenC00,Watrouslecturenote, Wilde13} for more details.

\begin{definition}\label{def:tracedistance}
	For quantum states $\rho$ and $\sigma$, the  $\ell_1$-distance between them is given by $\onenorm{\rho-\sigma}$, where $\onenorm{X}\defeq\Tr\sqrt{X^{\dag}X}$ is the sum of the singular values of $X$. We say that $\rho$ is $\varepsilon$-close to $\sigma$ if $\|\rho-\sigma\|_1\leq\varepsilon$.
\end{definition}
\begin{definition}\label{def:fidelity}
	For quantum states $\rho$ and $\sigma$, the {\em fidelity} between them is given by $\fidelity{\rho}{\sigma}\defeq\onenorm{\sqrt{\rho}\sqrt{\sigma}}.$ The {\em Hellinger distance} between them is defined as $h (\rho,\sigma)=\sqrt{1-\fidelity{\rho}{\sigma}}$. We also use $h\br{\rho,\atop\sigma}$ for overlong expressions.
\end{definition}

The following fact relates the $\ell_1$-distance and the fidelity between two states.

\begin{fact}[Fuchs-van de Graaf inequalities~\cite{FuchsG99}]
	\label{fac:tracefidelityequi}
	For quantum states $\rho$ and $\sigma$, it holds that
	\[2(1-\fidelity{\rho}{\sigma})=2h^2\br{\rho,\sigma}\leq\onenorm{\rho-\sigma}\leq2\sqrt{1-\fidelity{\rho}{\sigma}^2}.\]
	For pure states $\ket{\phi}$ and $\ket{\psi}$, we have
	\begin{align*}
	\onenorm{\ketbra{\phi} - \ketbra{\psi}}
	&= \sqrt{1 - \fidelity{\ketbra{\phi}}{\ketbra{\psi}}^2} \\
	&= \sqrt{1 - |\langle\phi|\psi\rangle|^2}.
	\end{align*}
\end{fact}

We use capital letters $A,B,\ldots$  to represent the registers; $\H_A,\H_B,\ldots$ to represent the Hilbert spaces associated to them and $\D_A,\D_B,\ldots$ to represent the set of all quantum states in $\H_A,\H_B,\ldots$. For any register $A$, $\abs{A}$ represents the number of qubits it contains, or equivalently, $\log\dim\H_A$. For bipartite $\rho_{AB}$, we define
\[ \rho_B \defeq \Tr_A\br{\rho_{AB}}
\defeq \sum_i (\bra{i} \otimes \id_{B})
\rho^{AB} (\ket{i} \otimes \id_{B}) \]
where $\set{\ket{i}}_i$ is a basis for the Hilbert space $\H_A$
and $\id_B$ is the identity matrix in space $\H_B$. $\Tr_A$ is called the partial trace operation.
The state $\rho_B$ is referred to as the marginal state of $\rho_{AB}$ in register $B$.
The following fact states that the distance between
two states can't be increased by quantum operations.
\begin{fact}
	\label{fac:monotonequantumoperation}
	For states $\rho$, $\sigma$, and quantum operation $\E(\cdot)$, it holds that
	\begin{align*}
	\onenorm{\E(\rho) - \E(\sigma)} &\leq \onenorm{\rho - \sigma}
	\intertext{and}
	\fidelity{\E(\rho)}{\E(\sigma)}&\geq \fidelity{\rho}{\sigma} .
	\end{align*}
	In particular, for any bipartite states $\rho_{AB}$ and $\sigma_{AB}$, it holds that
	\[\onenorm{\rho_{AB}-\sigma_{AB}}\geq\onenorm{\rho_A-\sigma_A}, \fidelity{\rho_{AB}}{\sigma_{AB}}\leq\fidelity{\rho_A}{\sigma_A}~\mbox{and}~ h\br{\rho_{AB},\sigma_{AB}}\geq h\br{\rho_A,\sigma_A}.\]
\end{fact}

\begin{fact}\label{fac:jointlylinearity}
	Given bipartite states $\rho_{AB}=\sum_{i}p_i\ketbra{i}\otimes\rho_i$ and $\sigma_{AB}=\sum_iq_i\ketbra{i}\otimes\sigma_i$, where $\set{p_i}_i$ and $\set{q_i}_i$ are distributions, it holds that
	\[\fidelity{\rho_{AB}}{\sigma_{AB}}=\sum_i\sqrt{p_iq_i}~\fidelity{\rho_i}{\sigma_i}.\]
\end{fact}

\begin{definition}\label{def:purification}
	We say that a pure state $\ket{\psi} \in \H_A\otimes\H_B$
	is a purification of some state $\rho$
	if $\Tr_A(\ketbra{\psi})=\rho$. If $\rho=\sum_ip\br{i}\ketbra{i}$ is a classical state, we say the {\em canonical purification} of $\rho$ is $\sum_i\sqrt{p\br{i}}\ket{i}\ket{i}$.
\end{definition}

\begin{fact}[Uhlmann's theorem]
	\label{fac:Uhlmann}
	Given quantum states $\rho$, $\sigma$, and a purification $\ket{\psi}$ of $\rho$,
	it holds that $\fidelity{\rho}{\sigma}=\max_{\ket{\phi}} | \langle \phi | \psi \rangle | $,
	where the maximum is taken over all purifications of $\sigma$. Let $\rho=\sum_i\alpha_i\ketbra{u_i}$ and $\sigma=\sum_i\beta_i\ketbra{v_i}$ be spectral decompositions of $\rho$ and $\sigma$, respectively; $\ket{\phi}_{AB}$ and $\ket{\psi}_{AB}$ be purifications of $\rho$ and $\sigma$, respectively, with Schmidt decomposition $\ket{\phi}=\sum_i\sqrt{\alpha_i}\ket{u_i}_A\ket{u_i'}_B$ and $\ket{\psi}=\sum_i\sqrt{\beta_i}\ket{v_i}_A\ket{v'_i}_B$. Let $\tilde{\rho},\tilde{\sigma}$ be marginals of $\ket{\phi},\ket{\psi}$ on register $B$ respectively. Let $U$ be the unitary such that $\sqrt{\tilde{\rho}}\sqrt{\tilde{\sigma}}U$ is positive semidefinite (guaranteed by the polar decomposition). Then $\bra{\phi}\br{\id_A\otimes U}\ket{\psi}=\fidelity{\rho}{\sigma}$. In particular, if $\tilde{\rho},\tilde{\sigma}$ are classical-quantum states, then $U$ can be assumed to be a controlled isometry on classical register.
\end{fact}

\begin{definition}\label{def:entropy}
The {\em entropy} of a quantum state $\rho$ (in register $X$) is defined as $\mathrm{S}(\rho) \defeq - \Tr\rho\log\rho.$
We also let $\mathrm{S}\br{X}_{\rho}$ represent $\mathrm{S}(\rho)$.
\end{definition}
\begin{definition}\label{def:relative entropy}
The {\em relative entropy} between quantum states $\rho$ and $\sigma$ is defined as $\relent{\rho}{\sigma} \defeq \Tr\rho\log\rho-\Tr\rho\log\sigma.$
\end{definition}
\begin{definition}\label{def:mutual entropy}
Let $\rho_{XY}$ be a quantum state in space $\H_X\otimes\H_Y$. The {\em mutual information} between registers $X$ and $Y$ is defined to be
\begin{align*}
\mutinf{X}{Y}_{\rho} &\defeq
\mathrm{S}\br{X}_{\rho}+\mathrm{S}\br{Y}_{\rho}-\mathrm{S}\br{XY}_{\rho}.
\end{align*}
\end{definition}
It holds that $\mutinf{X}{Y}_{\rho}
= \relent{\rho_{XY}}{\rho_{X} \otimes \rho_{Y}}$.

If $X$ is a classical register, namely
$\rho_{XY} = \sum_x \mu(x) \ketbra{x} \otimes \rho^x_Y$,
where $\mu$ is a probability distribution over $X$, then
\begin{align*}
\mutinf{X}{Y}_{\rho}
&= \mathrm{S}\br{Y}_{\rho}-\mathrm{S}\br{Y|X}_{\rho} \\
&= \mathrm{S}\br{\sum_x\mu(x)\rho^x_Y}-\sum_x\mu(x)\mathrm{S}\br{\rho^x_Y}
\end{align*}
where the {\em conditional entropy} is defined as
\[\mathrm{S}(Y|X)_{\rho}\defeq\expec{x\leftarrow\mu}{\mathrm{S}(\rho^x_Y)}.\]
For bipartite quantum state $\rho_{XY}$, $\mathrm{S}\br{XY}_{\rho}-\mathrm{S}\br{X}_{\rho}$ is not always nonnegative. For instance, $\mathrm{S}\br{XY}_{\rho}-\mathrm{S}\br{X}_{\rho}=-|X|$ if $\rho_{XY}$ is an EPR-state.

\begin{fact}\label{fac:arakilieb}\cite{ArakiL:1970}
	Given a bipartite state $\rho_{AB}$, it holds that
	\[\abs{\mathrm{S}\br{A}_{\rho}-\mathrm{S}\br{B}_{\rho}}\leq\mathrm{S}\br{AB}_{\rho}\leq\mathrm{S}\br{A}_{\rho}+\mathrm{S}\br{B}_{\rho}.\]
\end{fact}

Let $\rho_{XYZ}$ be a quantum state with $Y$ being a classical register.
The mutual information between $X$ and $Z$, conditioned on
$Y$, is defined as
\begin{align*}
\condmutinf{X}{Z}{Y}_{\rho} &\defeq
\expec{y \leftarrow Y}{\condmutinf{X}{Z}{Y=y}_{\rho}} \\
&= \mathrm{S}\br{X|Y}_{\rho} +
\mathrm{S}\br{Z|Y}_{\rho} -
\mathrm{S}\br{XZ|Y}_{\rho} .
\end{align*}
The following {\em chain rule} for mutual information follows
easily from the definitions, when $Y$ is a classical register.
\[ \mutinf{X}{YZ}_{\rho} = \mutinf{X}{Y}_{\rho} + \condmutinf{X}{Z}{Y}_{\rho} .\]
We will need the following basic facts.

\begin{fact}[\cite{Watrouslecturenote,Jain:2003:LBB:946243.946331}]
	\label{fact:pinkser}
	For quantum states $\rho$ and $\sigma$, it holds that
	\[ \onenorm{\rho-\sigma} \leq \sqrt{\relent{\rho}{\sigma}} \quad \text{ and } \quad  1-\fidelity{\rho}{\sigma}=h^2\br{\rho,\sigma}\leq\relent{\rho}{\sigma} . \]
\end{fact}

\begin{fact} \label{fact:mutinf is min}
	For quantum states $\rho_{XY}$, $\sigma_{X}$,
	and $\tau_{Y}$,
	it holds that
	\[ \relent{\rho_{XY}}{\sigma_{X} \otimes \tau_{Y}}
	\geq \relent{\rho_{XY}}{\rho_{X}\otimes \rho_{Y}}=\mutinf{X}{Y}_{\rho}. \]
	Combing with Fact~\ref{fact:pinkser}, it holds that
	\[h\br{\rho_{XY},\rho_X\otimes\rho_Y}\leq\sqrt{\mutinf{X}{Y}_{\rho}}.\]
\end{fact}

\begin{fact}
	\label{fact:subsystem monotone}
	Let $\rho$ and $\sigma$ be quantum states and $\E\br{\cdot}$ be a quantum channel. Then it holds that $$\relent{\rho}{\sigma} \geq \relent{\E\br{\rho}}{\E\br{\sigma}}.$$ Moreover, given a bipartite quantum state $\rho_{XY}$, let $\E^{Y\rightarrow Z}\br{\cdot}$ be a quantum operation on $Y$. Combining with Fact~\ref{fact:mutinf is min}, we find that
	\[\mutinf{X}{Y}_{\rho}\geq\mutinf{X}{Z}_{\E\br{\rho}}.\]
	If $\E\br{\cdot}$ is an isometry, then
	\[\mutinf{X}{Y}_{\rho}=\mutinf{X}{Z}_{\E\br{\rho}}.\]
\end{fact}
\begin{fact}\label{fact:dataprocessing}(Data-processing inequality). Given a tripartite quantum state $\rho_{XYB}$, where $XY$ are classical registers, with the property that $X$ is determined by $Y$, that is , $\mathrm{S}\br{X|Y}_{\rho}=0$. Then
		\[\mutinf{X}{B}_{\rho}\leq\mutinf{Y}{B}_{\rho}.\]
\end{fact}
\begin{fact}~\cite{Lieb:1973,LiebR:1973}(\textbf{Strong subadditivity theorem})\label{fact:strongsubadditivity}
	For any tripartite quantum state $\rho_{ABC}$, it holds that $\condmutinf{A}{C}{B}_{\rho}\geq0$.	
\end{fact}
\begin{fact}\label{lem:boundoncmi}
	Given a tripartite state $\rho_{ABC}$, it holds that
	$\condmutinf{A}{B}{C}_{\rho}\leq 2|B|.$
\end{fact}
\begin{lemma}\label{lem:decoupling}
	Consider a tripartite pure state $\ket{\psi}_{ABC}$ which satisfies $\mutinf{A}{C}\leq\epsilon$. Then for any purifications $\ket{\psi_1}_{AB_1}$ and $\ket{\psi_2}_{B_2C}$ of $\psi_A$ and $\psi_C$,respectively, there exists an isometry $U$ mapping $\H_B$ to $\H_{B_1}\otimes\H_{B_2}$ such that
	\[\abs{\bra{\psi_1}\bra{\psi_2}\br{\id_{AC}\otimes U_B}\ket{\psi}}\geq1-\epsilon.\]
	Combining with Fact~\ref{fac:tracefidelityequi}
	we have
	\[h\br{\br{\id_{AC}\otimes U_B}\psi\br{\id_{AC}\otimes U^{\dagger}_B},\psi_1\otimes\psi_2}\leq\sqrt{\epsilon}.\]
\end{lemma}
\begin{proof}
	From Fact~\ref{fact:pinkser} and Fact~\ref{fact:mutinf is min}, we have
	\[\fidelity{\psi_{AC}}{\psi_A\otimes\psi_C}\geq1-\epsilon.\]
	The conclusion now follows from Uhlmann's theorem and the Fuchs-van de Graaf inequalities (Facts~\ref{fac:tracefidelityequi} and~\ref{fac:Uhlmann}).
	\end{proof}

We need the following fact for state distribution.

\begin{fact}\label{fac:onewaycompression}\cite{Jain:2005:PEM:1068502.1068658}
	Given a target quantum state $\rho_{XAB}=\sum_xp\br{x}\ketbra{x}_X\otimes\rho^x_{AB}$, where the input register $X$ is held by Alice. There exists a one-way quantum protocol where Alice sends $\O\br{\br{\mutinf{X}{B}_{\rho}+1}/\delta^2}$ qubits to Bob such that $\expec{x\leftarrow p}{h^2\br{\rho^x_{AB},\tilde{\rho}^x_{AB}}}\leq\delta^2,$ where $\tilde{\rho}^x_{AB}$ is the state shared between Alice and Bob at the end of the protocol when the input is $x$. \footnote{In~\cite{Jain:2005:PEM:1068502.1068658}, the theorem is stated in terms of $\ell_1$ distance and the proof uses the quantum substate theorem~\cite{Jain2002}. Later, Jain and Nayak~\cite{JainNayak:2012} provided a simpler proof for the quantum substate theorem with better dependence on the parameters. With the strengthened quantum substate theorem, it is easy to verify that the compression in~\cite{{Jain:2005:PEM:1068502.1068658}} also has better dependence on the parameters in terms of Hellinger distance as stated in Fact~\ref{fac:onewaycompression}.}
\end{fact}
\subsection{Models of Quantum Communication Complexity}\label{sec:model}

Quantum communication complexity was introduced by Yao in~\cite{Yao:1993}. It studies the advantages and limitations of the players who are allowed to exchange quantum messages to accomplish a communication task. Here we describe two models of quantum communication complexity as follows.

\subsection*{Yao's Model}

The model we use here is slightly different from the original one defined in Yao~\cite{Yao:1993}. It is closer to the one of Cleve and Buhrman~\cite{CleveB:1997}, with pre-shared entanglement, but we allow the players to communicate with quantum messages.
In this model, an $r$-round protocol $\Pi$ for a given classical task from input registers $A_{in} = X$, $B_{in} = Y$ to output registers $A_{out}$, $B_{out}$ is
defined by a sequence of isometries $U_1$, $\cdots$, $U_{r + 1}$ along with a
pure state $\psi \in \D (T_A^{in}  T_B^{in})$ shared between Alice and Bob,
for arbitrary finite dimensional registers $T_A^{in}$, $T_B^{in}$: the pre-shared entanglement.
We need $r+1$ isometries in order to have $r$ messages since a first isometry is applied before the first message
is sent and a last one after the final message is received.
In  the case of even $r$, for appropriate finite
dimensional quantum memory registers $A_1$, $A_3$, $\cdots$, $A_{r - 1}$, $A^\prime$ held by Alice, $B_2$, $B_4$, $\cdots$, $B_{r - 2}$, $B^\prime$
held by Bob, and quantum communication registers $C_1$, $C_2$, $C_3$, $\cdots$, $C_r$
exchanged by Alice and Bob, we have
$U_1 \in \U(A_{in}  T_A^{in}, A_1  C_1)$,
$U_2 \in \U(B_{in}  T_B^{in}  C_1, B_2  C_2)$,
$U_3 \in \U(A_1  C_2, A_3  C_3)$,
$U_4 \in \U(B_2  C_3, B_4  C_4)$, $\cdots$ ,
$U_{r} \in \U(B_{r - 2}  C_{r - 1},   B_{out}  B^\prime  C_{r})$,
$U_{r  + 1} \in \U(A_{r - 1}  C_r, A_{out}  A^\prime)$, where $\U(A, B)$ is the set of unitary channels from $\H_A$ to $\H_B$ :
see Figure~\ref{fig:int_mod}.
We adopt the convention that, at the outset, $A_0 = A_{in} T_A^{in}$, $B_0 = B_{in} T_B^{in}$,
for odd $i$ with $1 \leq i < r$, $B_i = B_{i-1}$, for even $i$ with $1 < i \leq r$, $A_i = A_{i-1}$ and also $B_r = B_{r + 1} = B_{out} B^\prime$, and $A_{r+1} = A_{out} A^\prime$. In this way, after application of $U_i$, Alice holds register $A_i$, Bob holds register $B_i$ and the communication register is $C_i$.
In the case of an odd number of messages $r$, the registers corresponding to $U_r$, $U_{r+1}$ are changed accordingly.
We slightly abuse notation and also write $\Pi$ to denote the channel from registers $ A_{in} B_{in}$ to $A_{out} B_{out}$ implemented by the protocol, i.e.~for any
input distribution $\mu$ on $XY$ and $\rho_\mu$ encoding $\mu$ on input registers $A_{in} B_{in}$,
\begin{align}
\Pi (\rho_\mu) =
\Tr_{A^\prime B^\prime }{U_{r  + 1} U_r \cdots U_2 U_1 (\rho_\mu \otimes \psi)}.
\end{align}

Note that the $ A^\prime $ and $ B^\prime $ registers are the final memory registers that are being discarded at the end of the protocol by Alice and Bob, respectively.

Recall that for a given state,  all purifications are related by isometries on the purification registers. For classical input registers $XY$ distributed according to $\mu$, we consider a canonical purification $\ket{\rho_\mu}^{X R_X Y R_Y}$ of $\rho_\mu^{A_{in} B_{in}}$, with
\begin{align}
\ket{\rho_\mu}^{X R_X Y R_Y} = \sum_{x, y} \sqrt{\mu (x, y)} \ket{xxyy}^{X R_X Y R_Y}.
\end{align}
We then say that the purifying registers $R_X R_Y$ contain \emph{quantum copies} of $XY$.
We define the {\em global state} at round $i$ to be the state on $XR_XYR_YA_iB_iC_i$, which is a pure state. Then the global state at round $i$ is
\begin{align}
\rho_i^{X R_X Y R_Y A_i B_i C_i} = U_i \cdots U_1 (\rho^{X R_X Y R_Y} \otimes \psi^{T_A^{in} T_B^{in}})
\end{align}
Also, we require that the final marginal state $\Pi (\rho^{A_{in} B_{in} R_X R_Y})$
on $R_X R_Y A_{out} B_{out}$ is classical. We say that a protocol $\Pi$ solves a function $f$ with error $\epsilon$ with respect to input distribution $\mu$ if $\Pr_\mu [\Pi (x, y) \not= f(x, y)] \leq \epsilon$, and we say $\Pi$ solves $f$ with error $\epsilon$ if $\max_{(x, y)} \Pr[\Pi (x, y) \not= f(x, y)] \leq \epsilon$.

We also make use of the notion of a {\em control-isometry}: it is an isometry acting on a classical-quantum state that leaves the content of the classical register unchanged. Such a classical register is called a control-register. In Yao's model, we assume that all the isometries $U_1,\ldots, U_{r+1}$ are control-isometries with control-register being the inputs.

\begin{figure}
\begin{overpic}[width=1\textwidth]{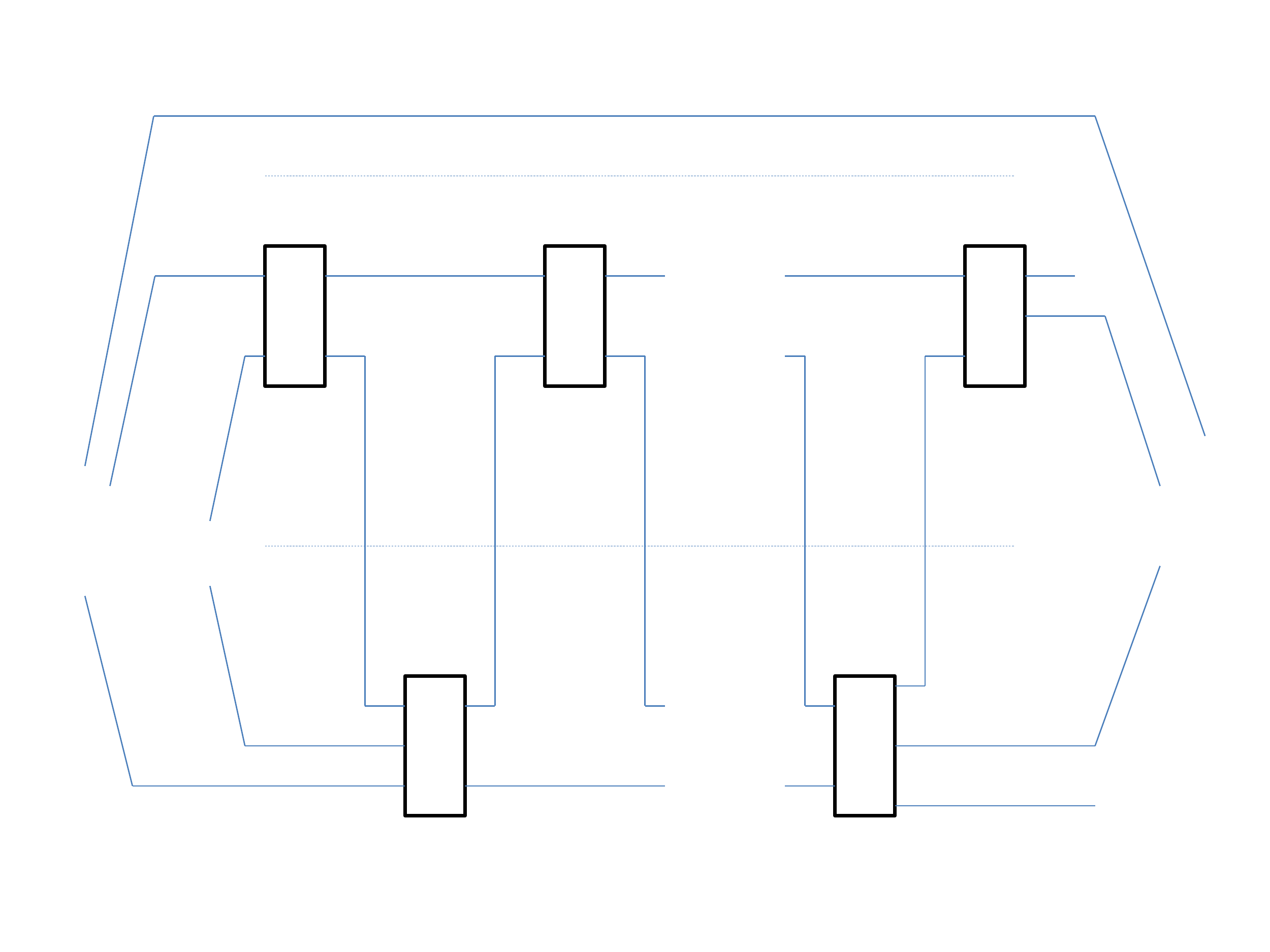}
  \put(0,65){Ray}
  \put(0,45){Alice}
  \put(0,22){Bob}
  \put(6,31){\footnotesize $\ket{\rho}$}
  \put(15,66.5){\footnotesize$R$}
  \put(15,54){\footnotesize$A_{in}$}
  \put(15,13.5){\footnotesize$B_{in}$}
  \put(15,44){\footnotesize$T_A^{in}$}
  \put(15,19){\footnotesize$T_B^{in}$}
  \put(22.1,49.5){\footnotesize$U_1$}
  \put(15,31){\footnotesize$\ket{\phi_1}$}
  \put(26.2,54){\footnotesize$A_1$}
  \put(26.2,48){\footnotesize$C_1$}
  \put(33,15.5){\footnotesize$U_2$}
  \put(37.2,54){\footnotesize$A_2$}
  \put(37.2,17.4){\footnotesize$C_2$}
  \put(37.2,13.7){\footnotesize$B_2$}
  \put(44.2,49.5){\footnotesize$U_3$}
  \put(48.2,54){\footnotesize$A_3$}
  \put(48.2,48){\footnotesize$C_3$}
  \put(48.2,13.7){\footnotesize$B_3$}
  \put(55,33){\footnotesize$\cdots$}
  \put(59.5,54){\footnotesize$A_{r-1}$}
  \put(59.5,48){\footnotesize$C_{r-1}$}
  \put(59.5,13.7){\footnotesize$B_{r-1}$}
  \put(66.5,15.5){\footnotesize$U_{r}$}
  \put(72.3,54){\footnotesize$A_r$}
  \put(73.5,23){\footnotesize$C_r$}
  \put(73.5,16.4){\footnotesize$B_{out}$}
  \put(73.5,11.7){\footnotesize$B^{\prime}$}
  \put(77.3,49.5){\footnotesize$U_{f}$}
  \put(81.3,54){\footnotesize$A^\prime$}
  \put(81.3,50.3){\footnotesize$A_{out}$}
  \put(92,33){\footnotesize$\Pi (\rho)$}
\end{overpic}
  \caption{Depiction of a quantum protocol in the interactive model, adapted from the long version of~\cite[Figure 1]{Touchette:2015}.}
  \label{fig:int_mod}
\end{figure}
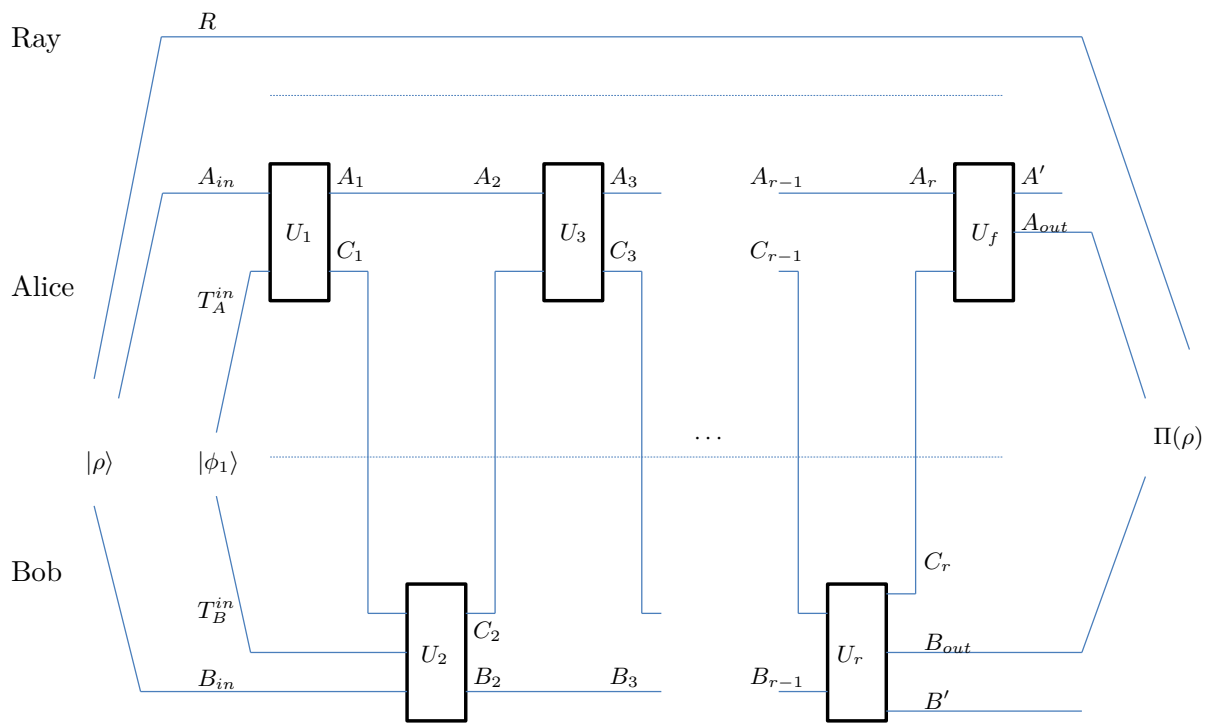

\subsection*{Cleve-Buhrman model}

In 1997, Cleve and Buhrman~\cite{CleveB:1997}
defined an alternative model for communication complexity in a quantum setting, in which
the players are allowed to pre-share an arbitrary entangled state but transmit classical rather
than quantum bits. This model is equivalent to Yao's model (with entanglement, up to a factor of 2), since
entanglement can be used to teleport~\cite{BennettBCJPW:1993} the qubits with twice as many classical
bits.

\subsection*{Quantum Communication Complexity and Quantum Information Complexity}

Since Yao's model (augmented with entanglement) and coherent Cleve-Buhrman model are equivalent up to factor 2, in this paper, we do not differentiate between these two models unless particularly specified.
\begin{definition}
For a protocol $\Pi$ and an input distribution $\mu$,
we define the {\em quantum communication cost} (QCC) and \emph{quantum information cost} (QIC) of $\Pi$ on input $\mu $ as
\begin{align*}
QCC (\Pi, \mu) &\defeq \sum_i |C_i|,
\end{align*}
and
\begin{align*}
	QIC (\Pi, \rho) &\defeq \sum_{i \geq 1,\ odd} \condmutinf{C_i}{ R_X R_Y}{ B_{i}} + \sum_{i \geq 1,\ even}  \condmutinf{C_i}{ R_X R_Y}{A_{i}},
\end{align*}
respectively.
For any function $f$, any input distribution $\mu$,  and any $\epsilon>0$,
\begin{align}\label{eq:defQCC}
QCC (f, \mu, \epsilon) \defeq \inf_{\Pi} QCC (\Pi, \mu),
\end{align}
and
	\begin{align}\label{eq:defQIComplexity}
		QIC (f, \mu, \epsilon) \defeq \inf_{\Pi} QIC (\Pi, \mu),
	\end{align}
	where the infimum is over the protocols $\Pi$ computing $f$ with error $\epsilon$ w.r.t $\mu$.
\end{definition}



\section{Lower bound on quantum communication complexity}
\label{sec:proofs}

In this section, we prove Theorem \ref{theo:qcclowerbound} by following the high-level proof sketch given in Section~\ref{sec:highlevel}.
We assume throughout this section that the protocol runs for $T$ rounds.
We first prove the Theorem assuming the results of the Lemmata in Section~\ref{sec:highlevel}, before proving these Lemmata.

\subsection{Proof of main Theorem}

\begin{proof}[\textbf{Proof of Theorem \ref{theo:qcclowerbound}}.]
Let $\Pi$ be a $T$-round quantum protocol with communication cost $c$. We assume without loss of generality that $t$ is odd and in the end of the protocol, Bob outputs the correct answer with probability at least $1-\epsilon>\frac{1}{2}$.

We first consider running protocol $\Pi$ on inputs given according to $p$.
We assume that the protocol is well-defined even outside the support of $\mu$, otherwise, adding an error flag as a potential output can only increase the distance of the output depending on whether $\Pi$ is run on $p$ or on $\mu$.
Let inputs to Alice and Bob be given in registers $XF$ and $YG$ in the state $$\sum_{x,y}p(x,y,f,g)\ketbra{x}_X\otimes\ketbra{y}_Y\otimes\ketbra{f}_F\otimes\ketbra{g}_G.$$ Let these registers be purified by $R_XR_F$ and $R_YR_G$ respectively, which are not accessible to either players. Let Alice and Bob initially hold registers $A_0,B_0$ with shared entanglement $\Theta^0_{T_AT_B}$. Then the initial state is
\begin{align*}\label{eq:initialstate}
\ket{\Psi^0}_{XYFGR_XR_YR_FR_GT_AT_B} \defeq \sum_{x,y,f,g}\sqrt{p(x,y,f,g)}\ket{xxyyffgg}_{XR_XYR_YFR_FGR_G}\ket{\Theta^0}_{T_AT_B}.
\end{align*}

Alice applies a control unitary $U^1: XFT_A\rightarrow XFA_1C_1$ such that the unitary acts on $T_A$ controlled by $XF$, then sends $C_1$ to Bob. Let $B_1\equiv T_B$ be a relabelling of Bob's register $B_0$. He applies $U^2: YGC_1B_1\rightarrow YGC_2B_2$ such that the unitary acts on $C_1B_0$ conditioned on $YG$. He sends $C_2$ to Alice. Players proceed in this fashion until the  end of the protocol. At round $r$, let the registers be $A_rC_rB_r$, where $C_r$ is the message register, $A_r$ is with Alice and $B_r$ is with Bob. If $r$ is odd, then $B_r \equiv B_{r-1}$ and if $r$ is even, then $A_r\equiv A_{r-1}$.
Then the global state at round $r$ is

\begin{equation*}
\label{eq:roundrstate}
\ket{\Psi^r}_{XYFGR_XR_YR_FR_GA_rC_rB_r} \defeq \sum_{x,y,f,g}\sqrt{p(x,y,f,g)}\ket{xxyyffgg}_{XR_XYR_YFR_FGR_G}\ket{\Theta^{r, xfyg}}_{A_rC_rB_r}.
\end{equation*}

Set $c_i\defeq|C_i|; \ell_{A,r}\defeq\sum_{i\leq r, i~\text{odd}}c_i$; $\ell_{B,r}\defeq\sum_{i\leq r, i~\text{even}}c_i$.
\begin{eqnarray*}
	\epsilon_{r,x_{\leq j}y_jj}&\defeq&h\br{\Psi^{r, x_{\leq j} jy_j}_{X_SF_SB_rY_{\geq j}\br{R_Y}_{\geq j}GR_G},\Psi^{r, x_{\leq j} jy_j}_{X_SF_S}\otimes\Psi^{r, x_{\leq j} jy_j}_{B_rY_{\geq j}\br{R_Y}_{\geq j}GR_G}}\\
	&=&	\sqrt{\expec{x_sf_s\leftarrow X_SF_S}{h^2\br{\Psi^{r, x_{\leq j} jy_jx_sf_s}_{B_rY_{\geq j}\br{R_Y}_{\geq j}GR_G},\Psi^{r, x_{\leq j} jy_j}_{B_rY_{\geq j}\br{R_Y}_{\geq j}GR_G}}}}	
\end{eqnarray*}
when $r$ is odd,
\begin{eqnarray*}
	\epsilon_{r,x_{\leq j}y_jj}&\defeq&h\br{\Psi^{r, x_{\leq j} jy_j}_{Y_SG_SA_rX_{\geq j}\br{R_X}_{\geq j}FR_F},\Psi^{r, x_{\leq j} jy_j}_{Y_SG_S}\otimes\Psi^{r, x_{\leq j} jy_j}_{A_rX_{\geq j}\br{R_X}_{\geq j}FR_F}}\\
	&=&\sqrt{\expec{y_sg_s\leftarrow Y_SG_S}{h^2\br{\Psi^{r, x_{\leq j} jy_jy_sg_s}_{A_rX_{\geq j}\br{R_X}_{\geq j}FR_F},\Psi^{r, x_{\leq j} jy_j}_{A_rX_{\geq j}\br{R_X}_{\geq j}FR_F}}}},
\end{eqnarray*}
when $r$ is even, where the equalities are from Fact~\ref{fac:jointlylinearity}.
By taking appropriate choices of input into protocol $\Pi$, we can combine Lemmata~\ref{lem:shearer},~\ref{lem:coherentJRS},~\ref{lem:multirddirectsum}, and prove  that on average under $p$, Bob's state is almost independent of $x_S f_S$, and Alice's state is almost independent of $y_S g_S$. We get the following claim.

\begin{claim}
	\label{lem:raolemma2}
	It holds that for all $r\leq t, 0<\delta<1$,
\begin{equation}\label{eqn:eps}
	\expec{x_{\leq j}y_jj}{\epsilon^2_{r,x_{\leq j}y_jj}}\leq\O\br{\frac{\ell_{A,r}}{k\delta^2}}+12\delta^2 + 18\frac{\ell_{B,r}2^{2\ell_{A,r}+4}}{n},
\end{equation}
when $r$ is odd, and
\begin{equation}\label{eqn:eps2}
	\expec{x_{\leq j}y_jj}{\epsilon^2_{r,x_{\leq j}y_jj}}\leq\O\br{\frac{\ell_{B,r}}{k\delta^2}}+12\delta^2+18\frac{\ell_{A,r}2^{2\ell_{B,r}+4}}{n},
\end{equation}
when $r$ is even.
\end{claim}

To go from distribution $p$ to distributions $\mu_0$ and $\mu_1$, we make yet another appropriate choice of the input into protocol $\Pi$, so that Lemma~\ref{lem:rdelim} can be used.
Let
\begin{align*}
	&\br{\Phi_b^{0,x_{\leq j}y_jj}}_{X_SR_{X_S}F_SR_{F_S}Y_SR_{Y_S}G_SR_{G_S}}\\
	&=\sum_{x_sf_sy_sg_s}\sqrt{\mu_b\br{x_s,f_s,y_s,g_s~|~x_{\leq j},y_j,j}}\ket{x_sx_sf_sf_sy_sy_sg_sg_s}_{X_SR_{X_S}F_SR_{F_S}Y_SR_{Y_S}G_SR_{G_S}}
\end{align*}
be canonical purifications, for $b\in\set{0,1}$, of the inputs $(x_sf_s,y_sg_s)$ restricted to $S$ and drawn under distribution $\mu_0$ and $\mu_1$, respectively. Also let
 $\Phi_b^{T,x_{\leq j}y_jj}$
 be the final states after running protocol $\Pi$ on inputs distributed according to $\mu_0$ and $\mu_1$, respectively. According to our assumption in the beginning of the proof, $T$ is odd and Bob outputs the answer.
We get the following claim.

\begin{claim}
	\label{lem:ptomu}

There exist registers $\hat{A},\hat{B}$, control isometries
\begin{align*}
V^T_{x_{\leq j}y_jj} & \in\U\br{X_SF_SX_{S^c}F_{S^c}R_{X_{S^c}}R_{F_{S^c}}A_T,X_SF_S\hat{A}}, \\
  V^{T-1}_{x_{\leq j}y_jj} & \in\U\br{Y_SG_SY_{S^c}G_{S^c}R_{Y_{S^c}}R_{G_{S^c}}B_T,Y_SG_S\hat{B}}
\end{align*}
 controlled by $X_SF_S$ and $Y_SG_S$, respectively, and a quantum state $\hat{\Psi}  \in\D_{\hat{A}\hat{B}}$ satisfying that
\begin{align*}
	&h\br{V^T_{x_{\leq j}y_jj}V^{T-1}_{x_{\leq j}y_jj}\br{\Phi_0^{T,x_{\leq j}y_jj}}, \br{\Phi_0^{0,x_{\leq j}y_jj}}_{X_SR_{X_S}F_SR_{F_S}Y_SR_{Y_S}G_SR_{G_S}}\otimes \hat{\Psi}_{\hat{A}\hat{B}}}\\
	&\leq\epsilon_{T,x_{\leq j}y_jj} + \epsilon_{T-1,x_{\leq j}y_jj} + 2 \sum_{r=1}^{T-2} \epsilon_{r,x_{\leq j}y_jj},
\end{align*}
and
\begin{align*}
	&h\br{V^T_{x_{\leq j}y_jj}V^{T-1}_{x_{\leq j}y_jj}\br{\Phi_1^{T,x_{\leq j}y_jj}}, \br{\Phi_1^{0,x_{\leq j}y_jj}}_{X_SR_{X_S}F_SR_{F_S}Y_SR_{Y_S}G_SR_{G_S}}\otimes \hat{\Psi} _{\hat{A}\hat{B}}}\\
	&\leq\epsilon_{T,x_{\leq j}y_jj} + \epsilon_{T-1,x_{\leq j}y_jj} + 2 \sum_{r=1}^{T-2} \epsilon_{r,x_{\leq j}y_jj}.
\end{align*}
\end{claim}
\bigskip

Using this claim, we proceed as follows. Note that $\br{\Phi_0^{0,x_{\leq j}y_jj}}_{Y_S G_S} = \br{\Phi_1^{0,x_{\leq j}y_jj}}_{Y_S G_S}$, so
\begin{align*}
&\Tr_{X_SR_{X_S}F_SR_{F_S}R_{Y_S}R_{G_S}}\br{\br{\Phi_0^{0,x_{\leq j}y_jj}}_{X_SR_{X_S}F_SR_{F_S}Y_SR_{Y_S}G_SR_{G_S}}\otimes \hat{\Psi} _{\hat{A}\hat{B}}}\\
&=\Tr_{X_SR_{X_S}F_SR_{F_S}R_{Y_S}R_{G_S}}\br{\br{\Phi_1^{0,x_{\leq j}y_jj}}_{X_SR_{X_S}F_SR_{F_S}Y_SR_{Y_S}G_SR_{G_S}}\otimes \hat{\Psi} _{\hat{A}\hat{B}}}.
\end{align*}
By triangle inequality and Fact~\ref{fac:monotonequantumoperation}, we have
\begin{eqnarray*}
	&&h\br{\Tr_{X_SR_{X_S}F_SR_{F_S}R_{Y_S}R_{G_S}}\br{V^T_{x_{\leq j}y_jj}V^{T-1}_{x_{\leq j}y_jj}\br{\Phi_0^{T,x_{\leq j}y_jj}}},\atop\Tr_{X_SR_{X_S}F_SR_{F_S}R_{Y_S}R_{G_S}}\br{V^T_{x_{\leq j}y_jj}V^{T-1}_{x_{\leq j}y_jj}\br{\Phi_1^{T,x_{\leq j}y_jj}}}}\\
	&\leq& 2\br{\epsilon_{T,x_{\leq j}y_jj} + \epsilon_{T-1,x_{\leq j}y_jj} + 2 \sum_{r=1}^{T-2} \epsilon_{r,x_{\leq j}y_jj}}
\end{eqnarray*}

Further taking expectation over $x_{\leq j}y_jj$, we have
\begin{align*}
	&\expec{x_{\leq j}y_jj}{h\br{\Tr_{X_SR_{X_S}F_SR_{F_S}R_{Y_S}R_{G_S}}\br{V^T_{x_{\leq j}y_jj}V^{T-1}_{x_{\leq j}y_jj}\br{\Phi_0^{T,x_{\leq j}y_jj}}},\atop\Tr_{X_SR_{X_S}F_SR_{F_S}R_{Y_S}R_{G_S}}\br{V^T_{x_{\leq j}y_jj}V^{T-1}_{x_{\leq j}y_jj}\br{\Phi_1^{T,x_{\leq j}y_jj}}}}} & \hspace{-10em}\\
	&\leq\expec{x_{\leq j}y_jj}{2\br{\epsilon_{T,x_{\leq j}y_jj} + \epsilon_{T-1,x_{\leq j}y_jj} + 2 \sum_{r=1}^{T-2} \epsilon_{r,x_{\leq j}y_jj}}} & \hspace{-10em}\\
	&\leq 4\sum_{r=1}^T\expec{x_{\leq j}y_jj}{\epsilon_{r,x_{\leq j}y_jj}} & \hspace{-10em}\\
	&\leq 4\expec{x_{\leq j}y_jj}{\sqrt{T\sum_{r=1}^T\epsilon_{r,x_{\leq j}y_jj}^2}}& \hspace{-10em}~\mbox{(Cauchy-Schwarz inequality)}\\
	&\leq 4\sqrt{T\sum_{r=1}^T\expec{x_{\leq j}y_jj}{\epsilon_{r,x_{\leq j}y_jj}^2}}& \hspace{-10em} ~\mbox{(Concavity of $\sqrt{x}$)}\\
	&\leq 4\sqrt{T\sum_{r=1}\br{\br{\O\br{\frac{\ell_{A,r}}{k\delta^2}}+12\delta^2 + 18\frac{\ell_{B,r}2^{2\ell_{A,r}+4}}{n}}+\br{\O\br{\frac{\ell_{B,r}}{k\delta^2}}+12\delta^2+18\frac{\ell_{A,r}2^{2\ell_{B,r}+4}}{n}}}}& \hspace{-10em} \\ & ~ \mbox{(by Eq.~\eqref{eqn:eps}.)} & \hspace{-10em}\\
	&\leq 4\sqrt{\O\br{\frac{T^2c}{k\delta^2}}+24\delta^2T^2+\frac{c2^{2c+10}}{n}} & \hspace{-10em}\\
	&\leq 4\sqrt{\O\br{\frac{c^3}{k\delta^2}}+24\delta^2c^2+\frac{c2^{2c+10}}{n}}& \hspace{-10em} ~\mbox{(because $T\leq c$)},	
\end{align*}
If $\frac{c}{24k}\geq 1$, then $c\geq\Omega\br{k}$. Otherwise, choose $\delta\defeq\br{\frac{c}{24k}}^{1/4}$. Then we have
\begin{eqnarray}
	&&\expec{x_{\leq j}y_jj}{h\br{\Tr_{X_SR_{X_S}F_SR_{F_S}R_{Y_S}R_{G_S}}\br{V^T_{x_{\leq j}y_jj}V^{T-1}_{x_{\leq j}y_jj}\br{\Phi_0^{T,x_{\leq j}y_jj}}},\atop\Tr_{X_SR_{X_S}F_SR_{F_S}R_{Y_S}R_{G_S}}\br{V^T_{x_{\leq j}y_jj}V^{T-1}_{x_{\leq j}y_jj}\br{\Phi_1^{T,x_{\leq j}y_jj}}}}}\nonumber\\
	&\leq&4\br{\O\br{\frac{c^5/2}{\sqrt{k}}}+\frac{c2^{2c+10}}{n}}.\label{eqn:distance}
\end{eqnarray}

On the other hand, we have
	\begin{eqnarray*}
		&&\expec{x_{\leq j}y_jj}{h\br{\Tr_{X_SR_{X_S}F_SR_{F_S}R_{Y_S}R_{G_S}}\br{V^T_{x_{\leq j}y_jj}V^{T-1}_{x_{\leq j}y_jj}\br{\Phi_0^{T,x_{\leq j}y_jj}}},\atop\Tr_{X_SR_{X_S}F_SR_{F_S}R_{Y_S}R_{G_S}}\br{V^T_{x_{\leq j}y_jj}V^{T-1}_{x_{\leq j}y_jj}\br{\Phi_1^{T,x_{\leq j}y_jj}}}}}\\
		&\geq&\expec{x_{\leq j}y_jj}{h\br{\Tr_{X_SR_{X_S}F_SR_{F_S}R_{Y_S}R_{G_S}\hat{A}}\br{V^T_{x_{\leq j}y_jj}V^{T-1}_{x_{\leq j}y_jj}\br{\Phi_0^{T,x_{\leq j}y_jj}}},\atop\Tr_{X_SR_{X_S}F_SR_{F_S}R_{Y_S}R_{G_S}\hat{A}}\br{V^T_{x_{\leq j}y_jj}V^{T-1}_{x_{\leq j}y_jj}\br{\Phi_1^{T,x_{\leq j}y_jj}}}}}\\
		&=&\expec{x_{\leq j}y_jj}{h\br{\Tr_{F_SR_{F_S}R_{G_S}A_T}\br{\br{\Phi_0^{T,x_{\leq j}y_jj}}}},\Tr_{F_SR_{F_S}R_{G_S}A_T}\br{\br{\Phi_1^{T,x_{\leq j}y_jj}}}}\geq \Omega\br{1}.
	\end{eqnarray*}	
 The equality is because $V^t_{x_{\leq j}y_jj}$and $V^{t-1}_{x_{\leq j}y_jj}$ are all control isometries controlled by $X_SF_S$ and $Y_SG_S$, respectively. The last inequality is because we assume that Bob outputs incorrect answers with probability at most constant $\varepsilon<\frac{1}{2}$.

Combining with~\eqref{eqn:distance}, we have
\begin{eqnarray*}
4\br{\O\br{\frac{c^{5/2}}{\sqrt{k}}}+\frac{c2^{2c+10}}{n}}\geq \Omega\br{1}.
\end{eqnarray*}
Therefore, quantum communication complexity for the communication task from Definition \ref{def:commtask} is at least
$\min\set{\Omega\br{k^{1/5}},\Omega\br{\log n}}$.
\end{proof}
	
\subsection{Proofs of the Claims}

In order to prove Claims~\ref{lem:raolemma2},~\ref{lem:ptomu}, we need the following two additional claims.

\begin{claim}\label{claim:boundoninformationloss}
	For any $r\leq t$, it holds that
	\begin{eqnarray*}
		&\condmutinf{XF}{C_rB_rYR_YGR_G}{JX_{<J}}\leq2\ell_{A,r}\\
		&\condmutinf{YG}{C_rA_rXR_XFR_F}{JY_{<J}}\leq2\ell_{B,r}
	\end{eqnarray*}
\end{claim}

\begin{proof}
	Let's prove the first inequality. The second one follows by symmetry. We prove it by induction on $r$.
	When $r=1$, any local operation on Bob's side does not increase $\condmutinf{XF}{C_1B_1YR_YGR_G}{JX_{<J}}$. We consider the state when Bob has received the first message and does not perform any operation. Then the left hand side is
	\[\condmutinf{XF}{C_1T_BYR_YGR_G}{JX_{<J}}=\condmutinf{XF}{C_1T_B}{JX_{<J}}=\condmutinf{XF}{C_1}{T_BJX_{<J}}\leq 2|C_1|=2\ell_{A,1}.\]
	The first equality is because $XF$ and $YG$ are independent conditioning on $JX_{<J}$. The second inequality is because $T_B$ is the part of the pre-shared entanglement, which is independent of the input. The last inequality is by Lemma~\ref{lem:boundoncmi}. If $r$ is odd, we have
	\begin{eqnarray*}
		&&\condmutinf{XF}{C_rB_rYR_YGR_G}{JX_{<J}}\\
		&\leq&\condmutinf{XF}{B_rYR_YGR_G}{JX_{<J}}+\condmutinf{XF}{C_r}{JX_{<J}B_rYR_YGR_G}\\
		&=&\condmutinf{XF}{B_{r-1}YR_YGR_G}{JX_{<J}}+2c_r\\
		&\leq&2\ell_{A,r},
	\end{eqnarray*}
	where the first inequality is  Lemma~\ref{lem:boundoncmi} and the second inequality is from the induction.
	If $r$ is even, we have
	\[\condmutinf{XF}{C_rB_rYR_YGR_G}{JX_{<J}}=\condmutinf{XF}{C_{r-1}B_{r-1}YR_YGR_G}{JX_{<J}},\]
	by Fact~\ref{fact:subsystem monotone}.
\end{proof}

\begin{claim}
	\label{lem:raolemma1}
	It holds that for all $r\leq t$,
	\begin{align}
		\condmutinf{X_J}{C_rB_rY_{>J}\br{R_Y}_{>J}GR_G}{Y_{\leq J}J} &
		\leq\frac{\ell_{A,r}2^{2\ell_{B,r}+2}}{n},\label{eqn:claim6.2 eq1}\\
		\condmutinf{Y_J}{C_rA_rX_{>J}\br{R_X}_{>J}FR_F}{X_{\leq J}J} &
		\leq\frac{\ell_{B,r}2^{2\ell_{A,r}+2}}{n}\label{eqn:claim6.2 eq2}.
	\end{align}

\end{claim}

\begin{proof}
	 With the notation from Lemma~\ref{lem:multirddirectsum}, taking $X \leftarrow X, J\leftarrow J, C\leftarrow C_r, B\leftarrow B_r, Y_1^J\leftarrow Y_{\geq J}G, Y_2^J\leftarrow Y_{<J}=X_{<J}$, we have

\begin{equation}\label{eqn:6.2}
   \condmutinf{X_J}{C_rB_rY_{\geq J}R_{Y_{\geq J}}GR_G}{JX_{<J}}\leq\frac{\ell_{A,r}2^{2\ell_{B,r}+2}}{n}.
\end{equation}
    Then
    \begin{align*}
    	&\condmutinf{X_J}{C_rB_rY_{>J}\br{R_Y}_{>J}GR_G}{Y_{\leq J}J} & \hspace{-2em}\\
    	&=\condmutinf{X_J}{C_rB_rY_{>J}\br{R_Y}_{>J}GR_G}{Y_JX_{<J}J} & \hspace{-2em} ~\mbox{(because $X_{<J}=Y_{<J}$)}\\
    	&=\condmutinf{X_J}{C_rB_rY_{\geq J}\br{R_Y}_{\geq J}GR_G}{X_{<J}J}& \hspace{-2em}  ~\mbox{($Y_J$ is independent of $X$ given $X_{<J}J$)}
    \end{align*}
    Together with Eq.~\eqref{eqn:6.2}, we have Eq.~\eqref{eqn:claim6.2 eq1}. Eq.~\eqref{eqn:claim6.2 eq2}. follows by the symmetric argument.

\end{proof}

\begin{proof}[\textbf{Proof of Claim \ref{lem:raolemma2}}]

	Let $r$ be odd, the case of even $r$ is proved similarly. We consider a new protocol $\Pi'$, where we have fixed $x_{\leq j}j$, and it is known to both Alice and Bob. The input to Alice is $XF$ ($X_{\leq j}$ is fixed) and  the input to Bob is $Y_j$. Note that $Y_j$ and $S$ determine each other given $x_{\leq j}j$.  Bob locally generates the registers $Y_{>j}\br{R_Y}_{>j}GR_G$, which are independent of $XF$ whenever $X_{\leq J}J$ is fixed. After that, Alice and Bob together simulate the original protocol $\Pi$ till round $r$. The global joint state is $\Psi^{r, x_{\leq j}j}$, which is the state $\Psi^r$ conditioned on fixing $x_{\leq j} j$.
	
	 $\mutinf{X_{>J}F}{B_rY_{>J}\br{R_Y}_{>J}GR_G}_{\Psi^{r, x_{\leq j}j}}\leq 2\ell_{A,r}$ by Claim~\ref{claim:boundoninformationloss}. As $\ell_{A,r}\geq 1$, from Lemma \ref{lem:coherentJRS}, there exists a one-way entanglement assisted protocol $\Pi''$ with communication cost $\O(\ell_{A,r}/\delta^2)$ and the final global state $\tilde{\Psi}^{x_{\leq j} j}_{X_{>j}\br{R_X}_{>j}FR_FA_rB_rY_{\geq j}\br{R_Y}_{\geq j}GR_G}$ such that
	\begin{equation}\label{eqn:oneway}
	h^2\br{\tilde{\Psi}^{x_{\leq j} j}_{X_{>j}FA_rB_rY_{\geq j}\br{R_Y}_{\geq} jGR_G},\Psi^{r,x_{\leq j} j}_{X_{>j}FA_rB_rY_{\geq j}\br{R_Y}_{\geq j}GR_G}}\leq\atop 4\delta^2+6\mutinf{Y_J}{A_rX_{>J}\br{R_X}_{>J}GR_G}_{\Psi^{r,x_{\leq j} j}}.
	\end{equation}

	As $Y_j$ and $S$ determine each other and $XF$ and $YG$ are independent for fixed $x_{\leq j}y_{\leq j}j$, we can further apply Lemma~\ref{lem:shearer} to obtain
	\begin{equation}\label{eqn:6.3}
		\condmutinf{X_SF_S}{B_rY_{\geq j}\br{R_Y}_{\geq j}GR_G}{S}_{\tilde{\Psi}^{x_{\leq j}y_{\leq j}j}}\leq\O\br{\frac{\ell_{A,r}}{k\delta^2}}.
	\end{equation}
	The reason is that $\tilde{\Psi}^{x_{\leq j}y_{\leq j}j}$ is obtained from the protocol $\Pi''$ with communication cost $\O\br{\ell_{A,r}/\delta^2}$. Combined with Claim~\ref{claim:boundoninformationloss},~\eqref{eqn:6.3} follows. By Fact~\ref{fact:mutinf is min}, we have
	\[\expec{s\leftarrow S}{h^2\br{\tilde{\Psi}^{x_{\leq j}jy_js}_{B_rY_{\geq j}\br{R_Y}_{\geq j}GR_GX_SF_S},\tilde{\Psi}^{x_{\leq j}jy_js}_{B_rY_{\geq j}\br{R_Y}_{\geq j}GR_G}\otimes\tilde{\Psi}^{x_{\leq j}jy_js}_{X_SF_S}}}\leq\O\br{\frac{\ell_{A,r}}{k\delta^2}}.\]
	It implies
	\begin{equation}\label{eqn:simlation}
		\expec{y_jsx_sf_s\leftarrow Y_jSX_SF_S}{h^2\br{\tilde{\Psi}^{x_{\leq j} jy_jsx_sf_s}_{B_rY_{\geq j}\br{R_Y}_{\geq j}GR_G},\tilde{\Psi}^{x_{\leq j} jy_js}_{B_rY_{\geq j}\br{R_Y}_{\geq j}GR_G}}}\leq\O\br{\frac{\ell_{A,r}}{k\delta^2}}.
	\end{equation}

	Combining~\eqref{eqn:oneway}~\eqref{eqn:simlation} and triangle inequality, we have
	\[\expec{y_jsx_sf_s\leftarrow Y_jSX_SF_S}{h^2\br{\Psi^{r,x_{\leq j} jy_jsx_sf_s}_{B_rY_{\geq j}\br{R_Y}_{\geq j}GR_G},\Psi^{r,x_{\leq j} jy_js}_{B_rY_{\geq j}\br{R_Y}_{\geq j}GR_G}}}\leq\atop\O\br{\frac{\ell_{A,r}}{k\delta^2}}+ 12 \delta^2 + 18\mutinf{Y_J}{A_rX_{>J}\br{R_X}_{>J}GR_G}_{\Psi^{r,x_{\leq j} j}}.\]
	Taking expectation over $x_{\leq j}j$, we have
	\begin{eqnarray*}
		&&\expec{jx_{\leq j}y_jsx_sf_s\leftarrow JX_{\leq J}Y_JSX_SF_S}{h^2\br{\Psi^{r,x_{\leq j} jy_jsx_sf_s}_{B_rY_{\geq j}\br{R_Y}_{\geq j}GR_G}-\Psi^{r,x_{\leq j} jy_js}_{B_rY_{\geq j}\br{R_Y}_{\geq j}GR_G}}}\\
		&\leq&\O\br{\frac{\ell_{A,r}}{k\delta^2}}+12\delta^2+18\expec{x_{\leq j}j\leftarrow X_{\leq J}J}{\mutinf{Y_J}{A_rX_{>J}\br{R_X}_{>J}GR_G}_{\Psi^{r,x_{\leq j} j}}}\\
		&=&\O\br{\frac{\ell_{A,r}}{k\delta^2}}+12\delta^2+18\condmutinf{Y_J}{A_rX_{>J}\br{R_X}_{>J}GR_G}{X_{\leq J}J}_{\Psi^r}\\
		&\leq&\O\br{\frac{\ell_{A,r}}{k\delta^2}}+12\delta^2+18\frac{\ell_{B,r}2^{2\ell_{A,r}+4}}{n},
	\end{eqnarray*}
	where the first equality is from the definition of $h\br{\cdot}$ and Fact~\ref{fac:jointlylinearity} and the last inequality is from Claim~\ref{lem:raolemma1}.

\end{proof}

\begin{proof}[\textbf{Proof of Claim \ref{lem:ptomu}}]

Consider the following communication task. Alice and Bob share $x_{\leq j}y_jj\leftarrow X_{\leq J}Y_JJ$. They are given $(x_sf_s,y_sg_s)$ as input. Note that $S$ is determined by $x_{\leq j}y_jj$. Moreover, since we are considering for now the fooling distribution $p$, $XF$ and $YG$ are independent when $x_{\leq j}y_jj$ is fixed. Alice and Bob locally sample the missing part of $XF$ and $YG$, respectively, and execute the protocol $\Pi$.
 Applying Lemma~\ref{lem:rdelim}, we get the claim.
\end{proof}

\subsection{Proofs of the Lemmata}

We now provide the proof of our lemmas.

\begin{proof}[\textbf{Proof of Lemma \ref{lem:shearer}}]
The result follows from the following chain of inequalities:

\begin{align*}
&\condmutinf{U_S}{V}{S} & \hspace{-20em}\\
&=\expec{s}{\sum_{i\in s}\condmutinf{U_i}{V}{U_{[<i]\cap s},S=s}} & \hspace{-20em}\\
&=\expec{s}{\sum_{i\in s}\condmutinf{U_i}{VU_{[<i]\setminus s}}{U_{[<i]\cap s}, S=s}-\condmutinf{U_i}{U_{[<i]\setminus s}}{VU_{[<i]\cap s},S=s}}& \hspace{-20em}\mbox{(Chain rule)}\\
&\leq\expec{s}{\sum_{i\in s}\condmutinf{U_i}{VU_{[<i]\setminus s}}{U_{[<i]\cap s}, S=s}} & \hspace{-20em}\mbox{(Fact~\ref{fact:strongsubadditivity})}\\
&=\expec{s}{\sum_{i\in s}\condmutinf{U_i}{V}{U_{<i},S=s}} & \hspace{-20em}\mbox{(Chain rule and independence of $U=U_1\otimes\ldots\otimes U_m$ and $S$)}\\
&=\expec{s}{\sum_{i\in s}\condmutinf{U_i}{V}{U_{<i}}}\hspace{6cm}& \hspace{-20em} \mbox{($UV$ is independent of $S$)}\\
&=\sum_i\prob{i\in S}\condmutinf{U_i}{V}{U_{<i}} & \hspace{-20em}\\
&\leq \frac{1}{k}\mutinf{U}{V}. & \hspace{-20em} \mbox{(Chain rule).}
\end{align*}

\end{proof}

\begin{proof}[\textbf{Proof of Lemma \ref{lem:coherentJRS}}]
	Note that $\mutinf{R_Y}{R_XXA}_{\Psi}=\mutinf{Y}{R_XXA}_{\Psi}$. By Lemma~\ref{lem:decoupling}, there exists a register $B'$ and an isometry $U_{YB}$ mapping $\H_Y\otimes\H_B$ to $\H_Y\otimes\H_{B'}$ such that
	\begin{equation}\label{eqn:4.2}
		h^2\br{U_{YB}\Psi\br{U_{YB}}^{\dag},\br{\sum_{y}\sqrt{\mu_Y\br{y}}\ket{yy}}\otimes\br{\sum_{y}\sqrt{\mu_Y\br{y}}\bra{yy}}_{YR_Y}\otimes\ketbra{\Phi}_{XR_XAB'}}\leq\epsilon,
	\end{equation}
	where $\ket{\Phi}$ is a purification of $\Psi_{XR_XYA}$.

Note that
	\[\mutinf{X}{B'}_{U_{YB}\Psi\br{U_{YB}}^{\dag}}\leq\mutinf{X}{YR_YB'}_{U_{YB}\Psi\br{U_{YB}}^{\dag}}=\mutinf{X}{YR_YB}_{\Psi},\] where the inequality is from Fact~\ref{fact:strongsubadditivity}. By Fact~\ref{fac:onewaycompression}, there exists a one-way quantum protocol, where Alice is given $x\sim\mu_X(x)$ and she sends $\O\br{\br{\mutinf{X}{YR_YB}_{\Psi}+1}/\delta^2}$ qubits to Bob such that
	\[h^2\br{\sum_x\mu_X\br{x}\ketbra{x}_X\otimes\tilde{\psi}^x_{AB'},\br{U_{YB}\Psi\br{U_{YB}}^{\dag}}_{XAB'}}\leq\delta^2,\]
	where $\tilde{\psi_x}$ is the shared state between Alice and Bob in the end of the protocol given input $x$.
	Combining with the previous inequality and Fact~\ref{fac:monotonequantumoperation}, we have
		\[h^2\br{\sum_x\mu_X\br{x}\ketbra{x}_X\otimes\tilde{\psi}^x_{AB'},\Phi_{XAB'}}\leq 2 \delta^2+ 2\epsilon.\]
	Hence
	\begin{equation}\label{eqn:4.22}
		h^2\br{\sum_x\mu_X\br{x}\ketbra{x}_X\otimes\br{\sum_{y}\sqrt{\mu_Y\br{y}}\ket{yy}}\otimes\br{\sum_{y}\sqrt{\mu_Y\br{y}}\bra{yy}}_{YR_Y}\otimes\tilde{\psi}^x_{AB'},\atop \br{\sum_{y}\sqrt{\mu_Y\br{y}}\ket{yy}}\otimes\br{\sum_{y}\sqrt{\mu_Y\br{y}}\bra{yy}}_{YR_Y}\otimes\Phi_{XAB'}}\leq 2 \delta^2+2 \epsilon.
	\end{equation}
	Combining ~\eqref{eqn:4.2}~\eqref{eqn:4.22}, Fact~\ref{fac:monotonequantumoperation} and triangle inequality, we have
	\[h^2 \br{\sum_x\mu_X\br{x}\ketbra{x}_X\otimes\br{\sum_{y}\sqrt{\mu_Y\br{y}}\ket{yy}}\otimes\br{\sum_{y}\sqrt{\mu_Y\br{y}}\bra{yy}}_{YR_Y}\otimes\tilde{\psi}^x_{AB'},\atop\br{U_{YB}\Psi\br{U_{YB}}^{\dag}}_{XYR_YAB'}}\leq 4 \delta^2 + 6\epsilon.\]
	Bob further applies $\br{U_{YB}}^{-1}$ on the registers $YB'$. He may need to extend the space $\H_B$ by adding the ancilla and trace out after applying $\br{U_{YB}}^{-1}$. By Fact~\ref{fac:monotonequantumoperation}, the Hellinger distance does not increase. We reach the desired conclusion.

\end{proof}

\begin{proof}[\textbf{Proof of Lemma \ref{lem:onerddirectsum}}]
Note that any quantum operation Bob performs does not increase\\ $\condmutinf{X_J}{C B Y_1^J R_{Y_1^J}}{ J X_{<J}}$ because of Fact~\ref{fact:subsystem monotone} and the assumption that $Y_2^J$ is a function of $X_{<J}$. It suffices to consider the global state when Bob receives $C$ and does not perform any operation.
Denote the global state by $\hat{\rho}$. It follows that
\begin{align*}
&\condmutinf{X_J}{C_1 B_1 Y_1^J R_{Y_1^J}}{ J X_{<J})}_{\rho} & \hspace{-12em} \\ &\leq\condmutinf{X_J}{C_1 T_B Y_1^J R_{Y_1^J}}{ J X_{<J})}_{\hat{\rho}} & \hspace{-12em} \mbox{($T_B$ is the marginal of the pre-shared states on Bob's side)} \\
& = \expec{j, x_{<j}}{ I (X_J ; C_1 T_B Y_1^J R_{Y_1^J} | J = j,  X_{<J} = x_{<j})}_{\hat{\rho}}  \\
& = \expec{j, x_{<j}}{ I (X_J ; C_1 T_B  | J = j,  X_{<J} = x_{<j})}_{\hat{\rho}} & \hspace{-12em}~\mbox{($Y_1^JR_{Y_1^J}$ is independent of $X_JC_1T_B$ given $jx_{<j}$.)} \\
	& = \expec{j, x_{<j}}{\condmutinf{X_j }{ C_1 T_B }{  X_{<j} = x_{<j})}}_{\hat{\rho}} & \hspace{-12em}~\mbox{$(XC_1T_B$ is independent of $J$.)} \\
	& = \frac{1}{n} \sum_j  \condmutinf {X_j }{ C_1 T_B }{  X_{<j} }_{\hat{\rho}} & \hspace{-12em} \\
	& = \frac{1}{n}  \mutinf {X }{ C_1 T_B  }_{\hat{\rho}}&~\mbox{(The chain rule of mutual information.)} & \hspace{-12em} \\
	& = \frac{1}{n}  \condmutinf {X }{ C_1 }{ T_B  }_{\hat{\rho}} & \hspace{-12em}~\mbox{($T_B$ is part of the pre-shared entangled state, independent of the input.)} \\
	& \leq\frac{2|C_1|}{n}\leq \frac{2\ell}{n}. & \hspace{-12em}~\mbox{(Lemma~\ref{lem:boundoncmi}.)}
\end{align*}
\end{proof}

\begin{proof}[\textbf{Proof of Lemma \ref{lem:multirddirectsum}}]
We assume the protocol $\Pi$ is in the Yao model defined in Section~\ref{sec:model}. Note that any local operation on Bob's side does not increase $\condmutinf{X_J }{C_rB_r Y_1^J R_{Y_1^J} }{ J X_{<J}}$. It suffices to consider the case that $r$ is odd.
We first convert $\Pi$ to a new protocol $\Pi'$ in a Cleve-Buhrman model using quantum teleportation, where the communication cost doubles. We construct a one-way protocol $\Pi''$ simulating the first $r$ rounds of $\Pi'$ as follows. For this, we use an additional register $P$ on Bob's side. It informs whether Bob aborts the protocol or not. In $\Pi''$, Alice and Bob share the entanglement state as in $\Pi$ and $c_2+c_4+\ldots+c_{r-1}$ copies of EPR states additionally.

For each odd round $t$, Alice measures $(c_{\frac{t-1}{2}}+1)$-th, $\ldots, c_{\frac{t+1}{2}}$-th copies of the EPR pairs on her side in computational basis and treats the outcome as the message Bob sent in round $t-1$. Alice performs exactly same as in $\Pi$. For each even round $t$, Bob measures the register $P$. If it is $1$, he does not perform any further operation. If it is $0$, Bob performs and prepares the message same as he is supposed to send in round $t$ of  $\Pi'$, denoted by $M_t$. Meanwhile, he also measures $(c_{\frac{t}{2}-1}+1)$-th, $\ldots, c_{\frac{t}{2}}$-th copies of the EPR pairs in computational basis. If the outcome is not same as $M_t$, he flips the bit in $P$ to $1$. Otherwise, he proceeds to the next round directly.  For protocol $\Pi''$, we define
$\Gamma_{XR_XYR_YACB}$ to be the global state when Bob has received message $C$ and does not perform any quantum operation; and $\tilde{\Theta}_{XYR_XR_Y\tilde{A}\tilde{B}P}$ to be the global state after Bob performs his quantum operation. Here we drop the superscript $r$ to simplify the notations. It is easy to see that $\prob{P=0}=2^{-2\ell_{B,r}}$.
Set
$$\Theta_{XYR_XR_Y\tilde{A}\tilde{B}}\defeq\tilde{\Theta}^{P=0}_{XYR_XR_Y\tilde{A}\tilde{B}},$$
$\Psi^r_{XYR_XR_YA_rB_rC_r}$ to be the global state of protocol $\Pi$ in round $r$.
We have
\begin{eqnarray}
&&\condmutinf{X_J }{B_rC_r Y_1^J R_{Y_1^J} }{ J X_{<J}}_{\Psi^r}\nonumber\\
&\leq&\condmutinf{X_J }{\tilde{B}Y_1^J R_{Y_1^J} }{ J X_{<J}}_{\Theta}\nonumber\\ &\leq&2^{2\ell_{B,r}} \condmutinf{X_J }{\tilde{B} Y_1^J R_{Y_1^J} }{ J X_{<J}P}_{\tilde{\Theta}}, \label{eqn:3}
\end{eqnarray}
where the first inequality is from the fact that the Cleve-Buhrman model obtained by quantum teleportation can be converted back to Yao's model via local quantum operations and Fact~\ref{fac:monotonequantumoperation}.
As $\Pi''$ is a one-way protocol, we have
\begin{align*}
&\condmutinf{X_J }{\tilde{B} Y_1^J R_{Y_1^J} }{ J X_{<J}P}_{\tilde{\Theta}} & \hspace{-5em}\\
&\leq \condmutinf{X_J }{\tilde{B} Y_1^J R_{Y_1^J}P }{ J X_{<J}}_{\tilde{\Theta}}& \hspace{-5em} ~\mbox{(Chain rule of the mutual information)}\\
&\leq \condmutinf{X_J}{CBY_1^J R_{Y_1^J}}{JX_{<J}}_{\Gamma}& \hspace{-5em} ~\mbox{(Fact~\ref{fact:subsystem monotone} and the fact that $Y_2^J$ is a function of $X_{<J}$)}\\
&=\condmutinf{X_J}{CT_B}{JX_{<J}}_{\Gamma} & \hspace{-5em} ~\mbox{(From the assumption $Y_1^J$ is independent of $X_{\geq J}$ given $JX_{<J}$)}
\end{align*}

By Lemma~\ref{lem:onerddirectsum}, we have
\[\condmutinf{X_J}{CT_B}{JX_{<J}}_{\Gamma}\leq\frac{4\ell_{A,r}}{n}\]

Combining with~\eqref{eqn:3}, we obtain that for $r=$ odd,
$$\condmutinf{X_J }{B_rC_r Y_1^J R_{Y_1^J} }{ J X_{<J}}_{\Psi^r}\leq  \frac{\ell_{A,r}2^{{2\ell_{B,r}+2}}}{n}.$$

\end{proof}

\begin{proof}[\textbf{Proof of Lemma \ref{lem:rdelim}}]

	We will show that for $i=1$, $\gamma_1 = \delta_1 = \epsilon_1$,
	and for $i > 1$,
	\begin{align*}
	\gamma_i &\leq \gamma_{i-1} + \epsilon_{i-2}+ \epsilon_i, \\
	\delta & \leq \delta_{i-1} + \epsilon_{i-2}+ \epsilon_i,
	\end{align*}
	from which the result follows.
	
	First, notice that $\rho^i_{R_X Y R_Y B_i C_i}~,~\rho^i_{R_X} \otimes \rho^i_{Y R_Y B_i C_i}, \rho^i_{R_Y X R_X A_i C_i}, \rho^i_{R_Y} \otimes \rho^i_{X R_X A_i C_i}$ are all classical-quantum states. By Fact~\ref{fac:Uhlmann}, we can assume that the $V^i$'s are control isometries controlled by $X$  for odd $i$ and controlled by $Y$ for even $i$ (note that $X=R_X$ and $Y=R_Y$).
	For $i=1$, we can rewrite
	\begin{align*}
	\epsilon_1 & = h \br{\rho^1_{R_X Y R_Y B_1 C_1}, \rho^1_{R_X} \otimes \rho^1_{Y R_Y B_1 C_1}} \\
	& = h \br{V^1 \br{ \rho^1_{X R_X Y R_Y A_1 B_1 C_1}},
		\rho_{X R_X} \otimes \rho^1_{ \tilde{X}_1 \tilde{R}_{X_1} Y R_Y \tilde{A}_1 B_1 C_1}} \\
	& = h \br{V^1 \br{ \rho^1_{X R_X A_1 B_1 C_1}} \otimes \rho_{Y R_Y},
		\rho_{X R_X} \otimes \rho_{Y R_Y} \otimes \rho^1_{ \tilde{X}_1 \tilde{R}_{X_1} \tilde{A}_1 B_1 C_1} }.
	\end{align*}
	
	For $\gamma_1$, we can further apply $V_0$ on both sides. Note that $V^0$ and $V^1$ commute, we have
	\begin{align*}
	\epsilon_1 & = h \br{V^0\br{V^1  \br{ \rho^1_{X R_X A_1 B_1 C_1}} \otimes \rho_{Y R_Y}}, 
		V^0\br{\rho_{X R_X} \otimes \rho_{Y R_Y} \otimes \rho^1_{ \tilde{X}_1 \tilde{R}_{X_1} \tilde{A}_1 B_1 C_1}}} \\
	& = h \br{V^1 V^0 \br{ \rho^1_{X R_X A_1 B_1 C_1} \otimes \rho_{Y R_Y}}, 
		\rho_{X R_X} \otimes \rho_{Y R_Y} \otimes \rho^1_{ \tilde{X}_1 \tilde{R}_{X_1} \tilde{Y}_0 \tilde{R}_{Y_0} \tilde{A}_1 \tilde{B}_1 C_1}} \\
	& = \gamma_1,
	\end{align*}
	where the second equality follows from the fact that $\rho^1_{ \tilde{X}_1 \tilde{R}_{X_1} \tilde{Y}_0 \tilde{R}_{Y_0} \tilde{A}_1 \tilde{B}_1 C_1}=\rho^1_{\tilde{Y}_0 \tilde{R}_{Y_0}}\otimes\rho^1_{ \tilde{X}_1 \tilde{R}_{X_1}\tilde{A}_1 \tilde{B}_1 C_1}$.
	
	For $\delta_1$, we instead get rid of the uncorrelated state $\rho_Y$
	before applying $V^{Y|X} = V^{Y|X}_{X \rightarrow X Y R_Y}$ acting as a
	control unitary on $X$ and such that
	$V^{Y|X} \br{ \rho_{X R_X} }
	= \sigma_{XY R_X R_Y}$,
	as well as $V^{Y|X} \br{ \rho^1_{X R_X A_1 B_1 C_1} }
	= \sigma^1_{X R_X Y R_Y A_1 B_1 C_1}$,
	and get by then applying $V_0$,
	\begin{align*}
	\epsilon_1 & = h \br{V^1 \br{ \rho^1_{X R_X A_1 B_1 C_1}},
		\rho_{X R_X} \otimes  \rho^1_{ \tilde{X}_1 \tilde{R}_{X_1} \tilde{A}_1 B_1 C_1} } \\
	& = h \br{V^{Y|X} V^1 \br{\rho^1_{X R_X A_1 B_1 C_1}},
		V^{Y|X} \br{ \rho_{X R_X}} \otimes \rho^1_{ \tilde{X}_1 \tilde{R}_{X_1} \tilde{A}_1 B_1 C_1} } \\
	& = h \br{ V^1  \br{ \sigma^1_{X R_X Y R_Y A_1 B_1 C_1}},
		\sigma_{X R_X Y R_Y}  \otimes \rho^1_{ \tilde{X}_1 R_{X_1} \tilde{A}_1 B_1 C_1} } \\
	&=  h \br{V^0\br{ V^1  \br{ \sigma^1_{X R_X Y R_Y A_1 B_1 C_1}}},
		V^0\br{\sigma_{X R_X Y R_Y}  \otimes \rho^1_{ \tilde{X}_1 R_{X_1} \tilde{A}_1 B_1 C_1}}}\\
	& = h \br{ V^1 V^0 \br{ \sigma^1_{X R_X Y R_Y A_1 B_1 C_1} },
		\sigma_{X R_X Y R_Y}  \otimes \rho^1_{ \tilde{X}_1 \tilde{R}_{X_1} \tilde{Y}_0 \tilde{R}_{Y_0} \tilde{A}_1 \tilde{B}_1 C_1} } \\
	& = \delta_1,
	\end{align*}
	where the third equality is from the fact that $V^1$ and $V^{Y|X}$ commute.
	
	For $i > 1$, we focus on even $i$; the case odd $i$ is proven similarly.
	Denote
	\begin{align*}
	U^i & = U^i_{Y B_{i-1} C_{i-1} \rightarrow Y B_i C_i}, \\
	U^{i-1} & = U^{i-1}_{X A_{i-2} C_{i-2} \rightarrow X A_{i-1} C_{i-1}}, \\
	\end{align*}
	the protocol unitaries.
	Then
	\begin{align}
	U^i \br{ \rho^{i-1}_{X R_X Y R_Y A_{i-1} B_{i-1} C_{i-1}}} & = \rho^i_{X R_X Y R_Y A_i B_i C_i}, \label{eqn:rhoi-1toi}\\
	U^{i-1} \br{ \rho^{i-2}_{X R_X Y R_Y A_{i-2} B_{i-2} C_{i-2}}} & = \rho^{i-1}_{X R_X Y R_Y A_{i-1} B_{i-1} C_{i-1}},\label{eqn:rhoi-2toi-1} \\
	U^i \br{\sigma^{i-1}_{X R_X Y R_Y A_{i-1} B_{i-1} C_{i-1}}} & = \sigma^i_{X R_X Y R_Y A_i B_i C_i}, \label{eqn:sigmai-1toi}\\
	U^{i-1} \br{ \sigma^{i-2}_{X \bar{R}_X Y \bar{R}_Y A_{i-2} B_{i-2} C_{i-2}} } & = \sigma^{i-1}_{X R_X Y R_Y A_{i-1} B_{i-1} C_{i-1}}\label{eqn:sigmai-2toi-1}.
	\end{align}
	
	For $\gamma_i$, we first reduce to $\gamma_{i-1}$ using the triangle inequality:
	\begin{align*}
	\gamma_i & = h \br{V^i V^{i-1} \br{ \rho^i_{X R_X Y R_Y A_i B_i C_i}},
		\rho_{X R_X} \otimes \rho_{Y R_Y} \otimes \rho^i_{ \tilde{X}_{i-1} \tilde{R}_{X_{i-1}} \tilde{Y}_{i} \tilde{R}_{Y_{i}} \tilde{A}_i \tilde{B}_i C_i}} \\
	& \leq h \br{V^i V^{i-1} \br{ \rho^i_{X R_X Y R_Y A_i B_i C_i} },\atop
		V^i U^i \br{V^{i-2}}^{\dagger}
		\br{ \rho_{X R_X} \otimes \rho_{Y R_Y} \otimes \rho^{i-1}_{\tilde{X}_{i-1} \tilde{R}_{X_{i-1}} \tilde{Y}_{i-2} \tilde{R}_{Y_{i-2}} \tilde{A}_{i-1} \tilde{B}_{i-1} C_{i-1}}}} \\
	& + h\br{V^i U^i \br{V^{i-2}}^{\dagger}
		\br{\rho_{X R_X} \otimes \rho_{Y R_Y} \otimes \rho^{i-1}_{\tilde{X}_{i-1} \tilde{R}_{X_{i-1}} \tilde{Y}_{i-2} \tilde{R}_{Y_{i-2}} \tilde{A}_{i-1} \tilde{B}_{i-1} C_{i-1}}},  \atop
		\rho_{X R_X} \otimes \rho_{Y R_Y} \otimes \rho^i_{ \tilde{X}_{i-1} \tilde{R}_{X_{i-1}} \tilde{Y}_{i} \tilde{R}_{Y_{i}} \tilde{A}_i \tilde{B}_i C_i}}.
	\end{align*}
	Indeed, by rearranging and using that $\br{U^i}^\dagger$, acting on Bob's side, and $V_{i-1}$, acting on Alice's side, commute, we get that the first term is equal to $\gamma_{i-1}$:
	\begin{align*}
	&  h \br{V^i V^{i-1} \br{ \rho^i_{X R_X Y R_Y A_i B_i C_i} },\atop
		V^i U^i \br{V^{i-2}}^{\dagger}
		\br{ \rho^X_{X R_X} \otimes \rho^Y_{Y R_Y} \otimes \rho^{i-1}_{\tilde{X}_{i-1} \tilde{R}_{X_{i-1}} \tilde{Y}_{i-2} \tilde{R}_{Y_{i-2}} \tilde{A}_{i-1} \tilde{B}_{i-1} C_{i-1}}}}\\
	&= h \br{V^{i-2} \br{U^i}^{\dagger} V^{i-1} \br{ \rho^i_{X R_X Y R_Y A_i B_i C_i} },\atop
	 \rho^X_{X R_X} \otimes \rho^Y_{Y R_Y} \otimes \rho^{i-1}_{\tilde{X}_{i-1} \tilde{R}_{X_{i-1}} \tilde{Y}_{i-2} \tilde{R}_{Y_{i-2}} \tilde{A}_{i-1} \tilde{B}_{i-1} C_{i-1}}}\\
	& = h \br{V^{i-2}V^{i-1}\br{ \rho^i_{X R_X Y R_Y A_i B_i C_i} },\atop
		 \rho^X_{X R_X} \otimes \rho^Y_{Y R_Y} \otimes \rho^{i-1}_{\tilde{X}_{i-1} \tilde{R}_{X_{i-1}} \tilde{Y}_{i-2} \tilde{R}_{Y_{i-2}} \tilde{A}_{i-1} \tilde{B}_{i-1} C_{i-1}}}\\
	&=h \br{V^{i-1}V^{i-2}\br{U^i}^{\dagger}\br{ \rho^{i-1}_{X R_X Y R_Y A_{i-1} B_{i-1} C_{i-1}} },\atop
		\rho^X_{X R_X} \otimes \rho^Y_{Y R_Y} \otimes \rho^{i-1}_{\tilde{X}_{i-1} \tilde{R}_{X_{i-1}} \tilde{Y}_{i-2} \tilde{R}_{Y_{i-2}} \tilde{A}_{i-1} \tilde{B}_{i-1} C_{i-1}}}=\gamma_{i-1},
	\end{align*}
	where the last equality is from Eq.~\eqref{eqn:rhoi-1toi} and the commutativity of $V^{i-1}$ and $V^{i-2}$.
	For the second term, we again use the triangle inequality to reduce to $\epsilon_i$:
	\begin{align}
	& h\br{V^i U^i \br{V^{i-2}}^{\dagger}
		\br{\rho_{X R_X} \otimes \rho_{Y R_Y} \otimes \rho^{i-1}_{\tilde{X}_{i-1} \tilde{R}_{X_{i-1}} \tilde{Y}_{i-2} \tilde{R}_{Y_{i-2}} \tilde{A}_{i-1} \tilde{B}_{i-1} C_{i-1}}}, \atop
		\quad \quad \quad \quad
		\rho_{X R_X} \otimes \rho_{Y R_Y} \otimes \rho^i_{ \tilde{X}_{i-1} \tilde{R}_{X_{i-1}} \tilde{Y}_{i} \tilde{R}_{Y_{i}} \tilde{A}_i \tilde{B}_i C_i}} \nonumber\\
	& \leq  h\br{V^i U^i \br{V^{i-2}}^{\dagger}
		\br{ \rho_{Y R_Y} \otimes \rho^{i-1}_{\tilde{X}_{i-1} \tilde{R}_{X_{i-1}} \tilde{Y}_{i-2} \tilde{R}_{Y_{i-2}} \tilde{A}_{i-1} \tilde{B}_{i-1} C_{i-1}}},\atop
		V^i
		\br{ \rho^{i}_{\tilde{X}_{i-1} \tilde{R}_{X_{i-1}} Y R_Y \tilde{A}_{i} B_{i} C_{i}}} } \nonumber\\
	& + h\br{ V^i
		\br{  \rho^{i}_{\tilde{X}_{i-1} \tilde{R}_{X_{i-1}} Y R_Y \tilde{A}_{i} B_{i} C_{i}} }, \nonumber
		\rho_{Y R_Y} \otimes \rho^i_{ \tilde{X}_{i-1} \tilde{R}_{X_{i-1}} \tilde{Y}_{i} \tilde{R}_{Y_{i}} \tilde{A}_i \tilde{B}_i C_i}}\nonumber\\
	&=h\br{V^i U^i \br{V^{i-2}}^{\dagger}
		\br{ \rho_{Y R_Y} \otimes \rho^{i-1}_{\tilde{X}_{i-1} \tilde{R}_{X_{i-1}} \tilde{Y}_{i-2} \tilde{R}_{Y_{i-2}} \tilde{A}_{i-1} \tilde{B}_{i-1} C_{i-1}}},\atop
		V^i
		\br{ \rho^{i}_{\tilde{X}_{i-1} \tilde{R}_{X_{i-1}} Y R_Y \tilde{A}_{i} B_{i} C_{i}}} }+\epsilon_i\label{eqn:VUV},
	\end{align}
	in which we also use the fact that $U^i$, $V^i$ and $V^{i-2}$ all act on Bob's side to get rid of the uncorrelated state $\rho_{X R_X}$. Notice that $\tilde{X}_{i-1} \tilde{R}_{X_{i-1}}=\tilde{X}_i \tilde{R}_{X_i}$ as $i$ is even.  For the first term, we use the fact that $V^{i-2}$, acting on Bob's side,  and $U^{i-1}$, acting on Alice's side, commute to go from $\rho^{i-1}$ to $\rho^{i-2}$, and find that it equals $\epsilon_{i-2}$:
	\begin{align}
 &  h\br{V^i U^i \br{V^{i-2}}^{\dagger}
 	\br{ \rho_{Y R_Y} \otimes \rho^{i-1}_{\tilde{X}_{i-1} \tilde{R}_{X_{i-1}} \tilde{Y}_{i-2} \tilde{R}_{Y_{i-2}} \tilde{A}_{i-1} \tilde{B}_{i-1} C_{i-1}}},\atop
 	V^i
 	\br{ \rho^{i}_{\tilde{X}_{i-1} \tilde{R}_{X_{i-1}} Y R_Y \tilde{A}_{i} B_{i} C_{i}}} } \nonumber\\
	& = h\br{ \rho_{Y R_Y} \otimes \rho^{i-1}_{\tilde{X}_{i-1} \tilde{R}_{X_{i-1}} \tilde{Y}_{i-2} \tilde{R}_{Y_{i-2}} \tilde{A}_{i-1} \tilde{B}_{i-1} C_{i-1}},\atop
		V^{i-2}\br{U^i}^{\dagger}
		\br{ \rho^{i}_{\tilde{X}_{i-1} \tilde{R}_{X_{i-1}} Y R_Y \tilde{A}_{i} B_{i} C_{i}}} }\nonumber\\
	& = h\br{ \rho_{Y R_Y} \otimes \rho^{i-1}_{\tilde{X}_{i-1} \tilde{R}_{X_{i-1}} \tilde{Y}_{i-2} \tilde{R}_{Y_{i-2}} \tilde{A}_{i-1} \tilde{B}_{i-1} C_{i-1}},\atop
		V^{i-2}
		\br{ \rho^{i-1}_{\tilde{X}_{i-1} \tilde{R}_{X_{i-1}} Y R_Y \tilde{A}_{i-1} B_{i-1} C_{i-1}}} }\nonumber\\
	& = h\br{ U^{i-1}\br{\rho_{Y R_Y} \otimes \rho^{i-2}_{\tilde{X}_{i-2} \tilde{R}_{X_{i-2}} \tilde{Y}_{i-2} \tilde{R}_{Y_{i-2}} \tilde{A}_{i-2} \tilde{B}_{i-2} C_{i-2}}},\atop
		V^{i-2}U^{i-1}
		\br{ \rho^{i-2}_{\tilde{X}_{i-2} \tilde{R}_{X_{i-2}} Y R_Y \tilde{A}_{i-2} B_{i-2} C_{i-2}}} }\nonumber\\
	&=h\br{ \rho_{Y R_Y} \otimes \rho^{i-2}_{\tilde{X}_{i-2} \tilde{R}_{X_{i-2}} \tilde{Y}_{i-2} \tilde{R}_{Y_{i-2}} \tilde{A}_{i-2} \tilde{B}_{i-2} C_{i-2}},\atop
		V^{i-2}
		\br{ \rho^{i-2}_{\tilde{X}_{i-2} \tilde{R}_{X_{i-2}} Y R_Y \tilde{A}_{i-2} B_{i-2} C_{i-2}}} }=\epsilon_{i-2}.\label{eqn:VUV2}
	\end{align}
	The bound on $\gamma_i$ follows by combining these.

	To handle $\delta_i$,  similarly to $V^{Y|X}$, we define
	$V^{X|Y} = V^{X|Y}_{Y \rightarrow Y X R_X}$ acting as a
	control unitary on $Y$ and such that
	$V^{X|Y} \br{ \rho_{Y R_Y} }
	= \sigma_{XY R_X R_Y}$.
	We first reduce $\delta_i$ to $\delta_{i-1}$ using the triangle inequality:
	\begin{align*}
	\delta_i & = h \br{V^i V^{i-1} \br{ \sigma^i_{X R_X Y R_Y A_i B_i C_i}},
		\sigma_{X R_X Y R_Y} \otimes \rho^i_{ \tilde{X}_{i-1} \tilde{R}_{X_{i-1}} \tilde{Y}_{i} \tilde{R}_{Y_{i}} \tilde{A}_i \tilde{B}_i C_i}} \\
	& \leq h \br{V^i V^{i-1} \br{ \sigma^i_{X R_X Y R_Y A_i B_i C_i} },\atop
		V^i U^i \br{V^{i-2}}^{\dagger}
		\br{ \sigma_{X R_X Y R_Y} \otimes \rho^{i-1}_{\tilde{X}_{i-1} \tilde{R}_{X_{i-1}} \tilde{Y}_{i-2} \tilde{R}_{Y_{i-2}} \tilde{A}_{i-1} \tilde{B}_{i-1} C_{i-1}}}} \\
	& + h\br{V^i U^i \br{V^{i-2}}^{\dagger}
		\br{\sigma_{X R_X Y R_Y} \otimes \rho^{i-1}_{\tilde{X}_{i-1} \tilde{R}_{X_{i-1}} \tilde{Y}_{i-2} \tilde{R}_{Y_{i-2}} \tilde{A}_{i-1} \tilde{B}_{i-1} C_{i-1}}}, \atop
		\sigma_{X R_X Y R_Y} \otimes \rho^i_{ \tilde{X}_{i-1} \tilde{R}_{X_{i-1}} \tilde{Y}_{i} \tilde{R}_{Y_{i}} \tilde{A}_i \tilde{B}_i C_i}}. \\
	\end{align*}
	Similarly to $\gamma_i$, we get that the first term is equal to $\delta_{i-1}$:
	\begin{align*}
	\delta_{i-1} & = h \br{V^i V^{i-1} \br{ \sigma^i_{X R_X Y R_Y A_i B_i C_i} },
		\br{ \sigma_{X R_X Y R_Y} \otimes \rho^{i-1}_{\tilde{X}_{i-1} \tilde{R}_{X_{i-1}} \tilde{Y}_{i-2} \tilde{R}_{Y_{i-2}} \tilde{A}_{i-1} \tilde{B}_{i-1} C_{i-1}}}}.
	\end{align*}
	For the second term, since $U_i$, $V_i$ and $V_{i-2}$ all act on Bob's side,  we apply $\br{V^{X | Y}}^\dagger$ on both side to get the same term as for $\gamma_i$, which was proved to be at most  $\epsilon_i + \epsilon_{i-2}$:
	\begin{align*}
	&h\br{ V^i U^i \br{V^{i-2}}^{\dagger}
		\br{\sigma_{X R_X Y R_Y} \otimes \rho^{i-1}_{\tilde{X}_{i-1} \tilde{R}_{X_{i-1}} \tilde{Y}_{i-2} \tilde{R}_{Y_{i-2}} \tilde{A}_{i-1} \tilde{B}_{i-1} C_{i-1}}}, \atop
		\sigma_{X R_X Y R_Y} \otimes \rho^i_{ \tilde{X}_{i-1} \tilde{R}_{X_{i-1}} \tilde{Y}_{i} \tilde{R}_{Y_{i}} \tilde{A}_i \tilde{B}_i C_i}}. \\
	& = h\br{\br{V^{Y | X}}^\dagger V^i U^i \br{V^{i-2}}^{\dagger}
		\br{\sigma_{X R_X Y R_Y} \otimes \rho^{i-1}_{\tilde{X}_{i-1} \tilde{R}_{X_{i-1}} \tilde{Y}_{i-2} \tilde{R}_{Y_{i-2}} \tilde{A}_{i-1} \tilde{B}_{i-1} C_{i-1}}}, \atop
		\br{V^{Y | X}}^\dagger \sigma_{X R_X Y R_Y} \otimes \rho^i_{ \tilde{X}_{i-1} \tilde{R}_{X_{i-1}} \tilde{Y}_{i} \tilde{R}_{Y_{i}} \tilde{A}_i \tilde{B}_i C_i}}. \\
	& =  h\br{V^i U^i \br{V^{i-2}}^{\dagger} 
		\br{\rho_{Y R_Y} \otimes \rho^{i-1}_{\tilde{X}_{i-1} \tilde{R}_{X_{i-1}} \tilde{Y}_{i-2} \tilde{R}_{Y_{i-2}} \tilde{A}_{i-1} \tilde{B}_{i-1} C_{i-1}}},
		\rho_{Y R_Y} \otimes \rho^i_{ \tilde{X}_{i-1} \tilde{R}_{X_{i-1}} \tilde{Y}_{i} \tilde{R}_{Y_{i}} \tilde{A}_i \tilde{B}_i C_i}} \\
	& =  h\br{V^i U^i \br{V^{i-2}}^{\dagger} 
		\br{\rho_{X R_X} \otimes \rho_{Y R_Y} \otimes \rho^{i-1}_{\tilde{X}_{i-1} \tilde{R}_{X_{i-1}} \tilde{Y}_{i-2} \tilde{R}_{Y_{i-2}} \tilde{A}_{i-1} \tilde{B}_{i-1} C_{i-1}}}, \atop
		\rho_{X R_X} \otimes \rho_{Y R_Y} \otimes \rho^i_{ \tilde{X}_{i-1} \tilde{R}_{X_{i-1}} \tilde{Y}_{i} \tilde{R}_{Y_{i}} \tilde{A}_i \tilde{B}_i C_i}} \\
	& \leq \epsilon_i + \epsilon_{i-2}~\hspace{4cm}\mbox{\br{\text{Eqs.\br{~\eqref{eqn:VUV}\eqref{eqn:VUV2}}}}}
	\end{align*}
	The bound on $\delta_i$ follows by combining these.

\end{proof}

\bibliographystyle{alpha}
\bibliography{reference}
\end{document}